\documentclass[10pt,journal,cspaper,compsoc]{IEEEtran}
\ifCLASSOPTIONcompsoc
\else
\fi
\ifCLASSINFOpdf
\else
\fi
\usepackage[ruled,linesnumbered]{algorithm2e}
\usepackage{subfig}
\usepackage{amsmath}
\usepackage{amssymb}
\usepackage{amsthm}
\usepackage{amsfonts}
\usepackage{color}
\usepackage{courier}
\usepackage{listings}
\lstdefinestyle{sysj}{
  belowcaptionskip=1\baselineskip,
  breaklines=true,
  xleftmargin=\parindent,
  language=java,
  keywordstyle=\bfseries\color{blue},
  commentstyle=\itshape,
  numbers=left,
  numbersep=2pt,
  showstringspaces=false,
  basicstyle=\scriptsize\ttfamily,
  identifierstyle=\color{black},
  stringstyle=\color{orange},
}
\usepackage{graphicx}
\renewcommand{\algorithmcfname}{ALGORITHM}
\SetAlFnt{\small}
\SetAlCapFnt{\small}
\SetAlCapNameFnt{\small}
\SetAlCapHSkip{0pt}
\IncMargin{-\parindent}

\newtheorem{definition}{Definition}
\newtheorem{lem}{Lemma}
\newtheorem*{rem}{Remark}
\newtheorem{thm}{Theorem}

\usepackage{paralist}
\begin{document}
%
\title{A unified framework for modeling and implementation of hybrid
  systems with synchronous controllers}
%
%
%
%

\author{Avinash~Malik and Partha~Roop
  \IEEEcompsocitemizethanks{\IEEEcompsocthanksitem A. Malik and P. Roop
    are with the Department of Electrical and Computer Engineering,
    University of Auckland, Auckland,
    NZ,\protect\\
    E-mail: avinash.malik@auckland.ac.nz, p.roop@auckland.ac.nz}
}

\IEEEcompsoctitleabstractindextext{%
\begin{abstract}

  This paper presents a novel approach to including
  \textit{non-instantaneous} discrete control transitions in the linear
  hybrid automaton approach to simulation and verification of hybrid
  control systems. In this paper we study the control of a continuously
  evolving analog plant using a controller programmed in a synchronous
  programming language. We provide extensions to the synchronous subset
  of the SystemJ programming language for modeling, implementation, and
  verification of such hybrid systems. We provide a sound
  \textit{rewrite} semantics that approximate the evolution of the
  continuous variables in the discrete domain inspired from the
  classical supervisory control theory. The resultant discrete time
  model can be verified using classical model-checking tools. Finally,
  we show that systems designed using our approach have a higher
  fidelity than the ones designed using the hybrid automaton approach.

\end{abstract}

}

\maketitle

\IEEEdisplaynotcompsoctitleabstractindextext

%
\IEEEpeerreviewmaketitle

\section{Introduction}
%
%

%
%
%
%

\IEEEPARstart{M}{odern} closed loop control systems consist of a
physical process (termed the plant) controlled by a discrete embedded
controller. The plant is a continuously evolving (analog) system, which
is sampled by an analog to digital converter at specific
intervals. These samples are then input into the discrete controller,
which makes decisions depending upon the control logic and feeds the
resultant outputs back to the plant to control it.  The continuous time
nature of the plant and the discrete time nature of the controller
together form a hybrid system.  The \textit{Linear Hybrid
  Automaton}~\cite{Henzinger:1996:THA:788018.788803} is arguably the
most popular approach for modeling such hybrid systems. A linear hybrid
automaton captures the continuous evolution of the plant model as first
order ordinary differential equations (ODEs). In every \textit{control
  mode} of the discrete controller, the plant variables evolve according
to a set of ODEs, until an invariant condition holds. As soon as the
invariant condition is violated, an instantaneous switch is made by the
controller to a different control mode. The continuous variables in the
plant model can now evolve with a new set of ODEs. Thus, the controller
changes the plant behavior through this mode switch.

Control systems are reactive systems~\cite{Harel:1987:SVF:34884.34886}
that have an ongoing interaction with their respective plant in terms of
discrete time steps. At the start of each time step, the inputs from the
plant are captured, a reaction function is called to process these
inputs, and finally the outputs are emitted back to the
plant. Synchronous languages such as Esterel~\cite{berry96},
Lustre~\cite{nhal91}, Signal~\cite{pgue91} are used extensively to
implement such reactive systems, since synchronous programs can be
translated into transition systems in polynomial time even with
exponentially large number of states. Furthermore, model-checking of
temporal logic specifications~\cite{clarke-book00} can be directly
performed on these resultant symbolic transition systems to
\textit{guarantee} functional and real-time properties of the
controller. Synchronous languages, operate based on the principle of
\textit{synchrony hypothesis}, which requires that the reaction function
takes \textit{zero} time and the outputs are produced instantaneously.

Given the instantaneous mode switch of the hybrid automaton and the zero
delay computation model of the synchronous languages; it should
\textit{not} be surprising then that controllers \textit{modeled} in
hybrid automaton should be \textit{implemented} with synchronous
languages since semantically, the discrete step: mode switch and the
reaction function in both models takes zero time. However, in a real
system no controller takes zero time. The synchronous language community
has addressed this problem by considering the worst case reaction time
(WCRT) of a synchronous program~\cite{boldt07}. For a synchronous
controller; the WCRT, which is akin to the critical path of a program,
determines the inter-arrival time of input events. Statically obtaining
a tight WCRT for synchronous controllers is a well studied
problem~\cite{proop10,wilhelm08,boldt07}. To the best of our knowledge
an equivalent approach to incorporating time-delayed mode switches has
not been addressed by the hybrid automaton community. Consequently, any
results obtained from a system modeled using a hybrid automaton has low
fidelity, i.e., does not behave as expected due delays in making control
decisions.

In this paper \textbf{our main contribution} is: \textit{a powerful
  language with a precise formal semantics that allows the modeling,
  verification and implementation of non-trivial synchronous controllers
  with time-delays within their continuous environment.}  Our
contributions can be refined as follows:

\begin{compactitem}
\item Automatic, compiler driven, symbolic representation of the hybrid
  systems designed in the proposed hybrid synchronous language called
  \textit{HySysJ}.
\item A precise formal and novel natural semantics for compilation of
  hybrid systems.
\item The discrete approximation of hybrid system designed in HySysJ
  based on discrete linear time invariant systems from classical
  supervisory control theory~\cite{Astrom:2008:FSI:1816978}.
\end{compactitem}

Rest of the paper is arranged as follows: Section~\ref{sec:related-work}
gives a detailed description of the current state of the art in hybrid
system design, highlighting the
deficiencies. Section~\ref{sec:preliminaries} introduces the
preliminaries required to read the rest of the paper. We motivate the
problem using an example in Section~\ref{sec:motivating-example}. The
basic language definition is provided in
Section~\ref{sec:base-language}, which is further extended with
continuous time constructs and semantics in
Section~\ref{sec:example-driv-inform}. The relation of the proposed
approach to classical supervisory control theory is presented in
Section~\ref{sec:determ-value-wcrt}.
The verification procedure carried out on the motivating example in the
resultant new language is given in
Section~\ref{sec:manuf-syst-revis}. Finally, we conclude in
Section~\ref{sec:conclusions}.

\section{Related work}
\label{sec:related-work}

Many languages have been proposed for modeling and verification of
Hybrid systems. A good survey can be found
in~\cite{Carloni:2006:LTH:1166403.1166404}. The first class of languages
are the hardware description languages enhanced with the analog mixed
signal (AMS) extensions, such as; VHDL-AMS and
SystemC-AMS~\cite{1205169,Pecheux:2006:VVA:2298529.2301129}. These
languages lack any sort of formal semantics and hence, cannot be used
for formal verification. The second class is the data-flow languages
such as Z\'elus and
SCADE/Lustre~\cite{DBLP:conf/hybrid/BourkeP13,nhal91}, which approximate
the continuous ODEs. This approach of approximating the continuous ODE
behavior is essential, because model-checking most system properties,
including safety properties, are known to be undecidable for general
hybrid systems~\cite{Henzinger:1996:THA:788018.788803}. The
aforementioned data-flow languages are also endowed with formal
mathematical semantics.  This conjunction of approximation of continuous
behavior along with formal mathematical foundations makes these
languages potentially suitable for model-checking. But, \textit{unlike}
us, the overall hybrid model does not account for the non-zero
mode-switch times and hence, these programming languages suffer from the
same problems as the hybrid automata.

Finally, the work closest to the one described in this article is done
by: (1) Closse et al.~\cite{bertin2001taxys}, where they extend the
Esterel language to model timed automata~\cite{alur94}, i.e., ODEs with
rate of change always equal to 1. In this proposal we are able to model
the more general hybrid rather than its subset timed automaton and (2)
Baldamus et al.~\cite{baldamus2002modifying}, which is a seminal work in
extending synchronous imperative languages to model hybrid
automaton. This work is extended further and completed by giving a
formal treatment by Bauer et al.~\cite{bauer2010synchronous}. The work
described herein differs significantly from
both;~\cite{bauer2010synchronous} and~\cite{baldamus2002modifying} in
that they do not approximate the continuous behavior of the plant,
instead all discrete transitions are carried out and then a so called
continuous phase is launched, which models the continuous evolution of
the plant until the invariant condition holds, just like in hybrid
automaton. Since these approaches derive their semantics from hybrid
automaton, they inherit the same problem described in
Section~\ref{sec:motivating-example}, i.e., non-zero mode-switch
transitions cannot be captured in the semantics.

Overall, the formal foundations of the modeling/implementation language
proposed in this paper are truly unique, since the semantics unify the
real-time analysis of synchronous programs~\cite{boldt07} and the hybrid
modeling languages into a single framework inspired from classical
supervisory control theory.

\section{Preliminaries}
\label{sec:preliminaries}

In this section we give the background information required for
understanding the rest of the paper.

\subsection{The hybrid automaton}
\label{sec:hybrid-automaton}

We use the definition of linear Hybrid automaton
from~\cite{Henzinger:1996:THA:788018.788803}.

\begin{definition} A \textit{hybrid automaton} \textrm{H} is a tuple
  $(Loc,Var,Con,Lab,Edge,Act,Inv,Init)$ where
  \begin{compactitem}
    
  \item $(Loc,Var,Con,Lab,Edge,Act,Inv,Init)$ is a labelled transition
    system with $Loc$ a finite set of locations, real-valued variables
    $Var$, $V$ the set of valuation $v: Var \rightarrow \mathbb{R}$, and
    $\Sigma = Loc \times V$ the set of states, $Init \subseteq \Sigma$
    of initial states.
    
  \item A function $Con: Loc \rightarrow 2^{Var}$ assigning a set of
    controlled variables to each location
    
  \item a finite set of labels $Lab$, including the stutter label
    $\tau \in Lab$.

  \item $Act$ (Activities) is a function assigning a set of activities
    \mbox{$f: \mathbb{R}_{\geq 0} \rightarrow V$} to each location
    ($l \in Loc$) which are time-invariant meaning that $f \in Act(l)$
    implies \mbox{$(f+t) \in Act(l)$} where
    \mbox{$(f+t)(t') = f(t+t'), \forall t' \in \mathbb{R}_{\geq 0}$}
    
  \item a function $Inv$ assigning an invariant $Inv(l) \subseteq V$ to
    each location $l \in Loc$.
    
  \item A finite set
    $Edge \subseteq Loc \times Lab \times 2^{V^2} \times Loc$ of edges
    including $\tau$-transitions $(l,\tau,Id,l)$ for each location
    $l\in Loc$ with $Id =\{(v,v')|\forall x \in Con(l).v'(x) = v(x)\}$,
    and where all edges with label $\tau$ are $\tau$-transitions.

  \end{compactitem}
\end{definition}


\begin{definition}
  The semantics of a hybrid automaton $H$ is given by the operational
  semantics consisting of two rules, one for discrete instantaneous
  transition steps and one for continuous time steps.
  
  \begin{compactitem}
   
  \item Discrete step semantics (mode-switch semantics):
    \[
    \frac{(l,a,(v,v'),l') \in Edge\ \  v' \in Inv(l')}{(l,v) \xrightarrow{a} (l',v')}
    \]
  
\item Time step semantics
    \[
    \frac{f\in Act(l)\ f(0)=v\ f(t)=v'\ t\geq 0\ f([0,t])
      \subseteq Inv(l))}{(l,v) \xrightarrow{t} (l',v')}
    \]
  \end{compactitem}
\end{definition}

An execution step $\rightarrow = \xrightarrow{a} \cup \xrightarrow{t}$
of H is either a discrete step or a time step. A path $\pi$ of H is a
sequence $\sigma_0 \rightarrow \sigma_1 \ldots$ with
$\sigma_0 = (l_0,v_0) \in Init$, $v_0 \in Inv(l_0)$, and
$\sigma_0 \rightarrow \sigma_{i+1} \forall i \geq 0$.

\subsubsection{An example linear hybrid automaton}
\label{sec:an-example-linear-1}

We will use a closed loop manufacturing system example shown in
Figure~\ref{fig:1} to elaborate the semantics of hybrid automata.

\begin{figure}[t!]
  \centering
  \includegraphics[scale=0.6]{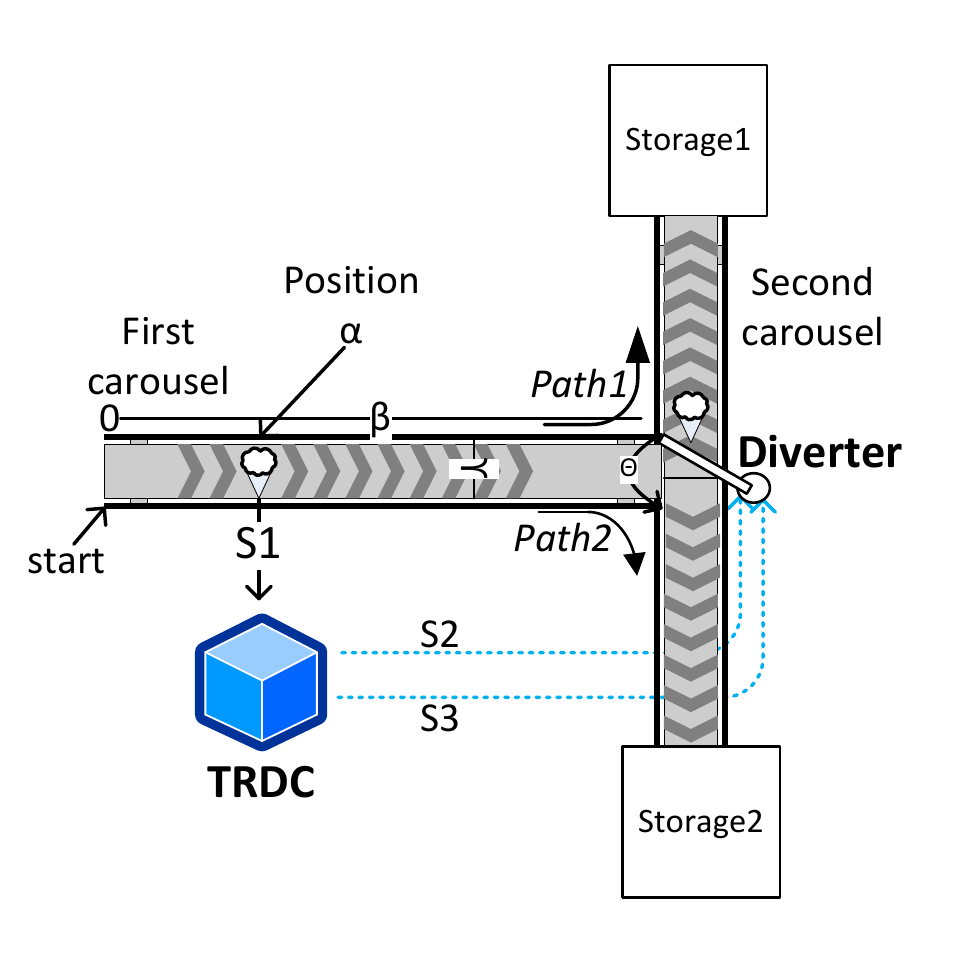} 
  \caption{The pictorial representation of the manufacturing
  control system}
  \label{fig:1}
\end{figure}

\begin{figure*}[t!]
  \centering
  \subfloat[The hybrid automaton modelling the manufacturing control
  system]{\scalebox{0.5}{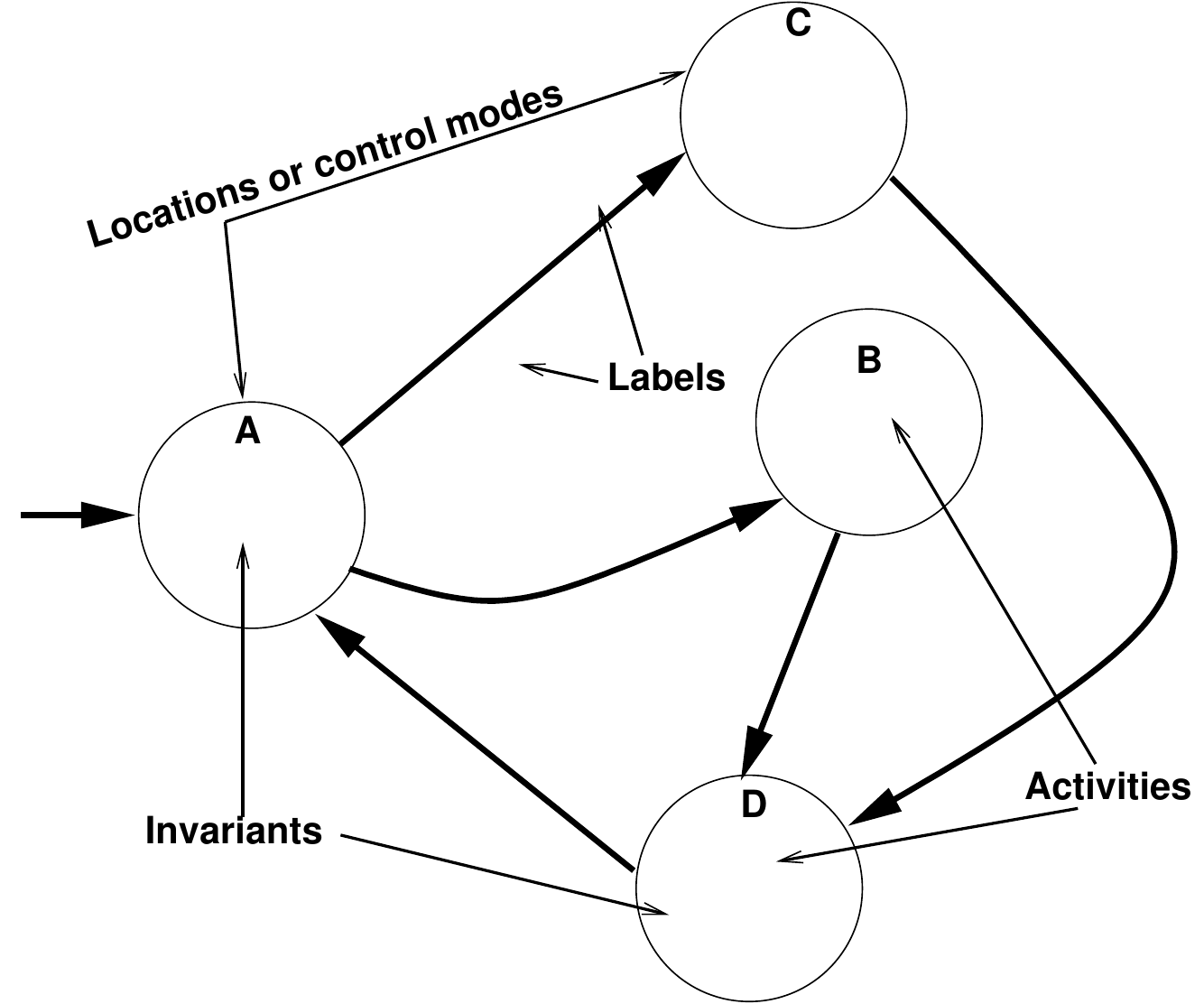 \label{fig:1b}}}
  \qquad
  \subfloat[The behavior of the manufacturing control system as modeled
  by the hybrid automaton]{
    \scalebox{0.6}{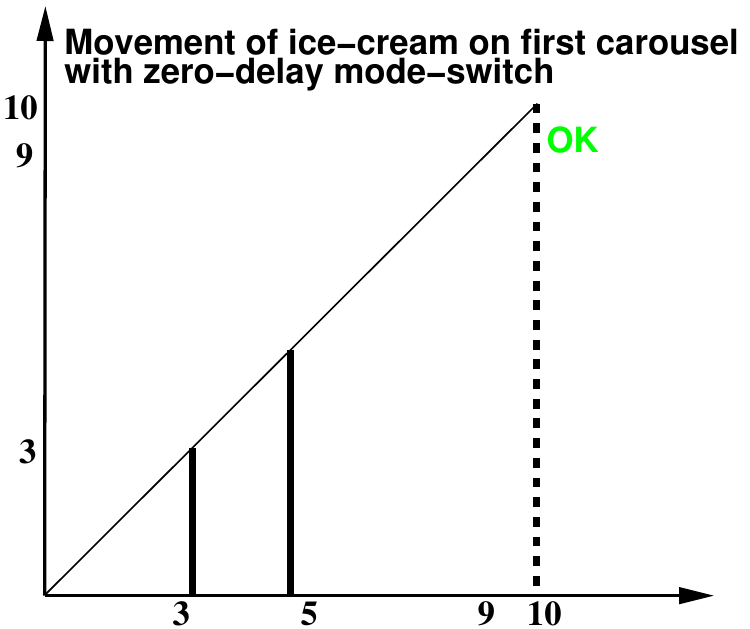}
    \label{fig:1c}
  }
  \caption{A simple carousel control system and its hybrid automaton}
  \label{fig:1n}
\end{figure*}

Consider that we are designing an automated ice-cream manufacturing
system as shown in Figure~\ref{fig:1}. The system consists of two
carousel belts that carry an ice-cream to either \texttt{Storage1} or
\texttt{Storage2} depending upon the RFID tag on the ice-cream. The size
of the first carousel is $\beta \times \gamma$ units. A diverter is
placed at the end of the first carousel, $\beta$ units from the start. A
tag reader and diverter controller (\texttt{TRDC}) is placed at position
$\alpha$ from the start of the first carousel. When the ice-cream is
detected, the \texttt{TRDC} reads the tag on the ice-cream and then
sends a control message to the diverter in order to move it into the
correct position, so that once the ice-cream reaches position $\beta$,
it is diverted to the correct storage station. The detection of the
ice-cream on the first carousel is indicated by the emission of signal
$S1$. Signal $S2$, emitted from the \texttt{TRDC}, moves the diverter
$\theta$ arc-length units in order to divert the ice-cream to
\texttt{Storage1}, while signal $S3$ does the opposite. Furthermore, the
carousel and the diverter move at a constant velocity of 1.

The hybrid automaton modeling the manufacturing system is shown in
Figure~\ref{fig:1b}. The elements of the tuple defining the syntax of
the hybrid automaton are indicated in Figure~\ref{fig:1b} for sake of
understanding. Initially, the ice-cream and the diverter are at position
0, denoted by the continuous variables $x$ and $y$, respectively. In
mode \textbf{A}, the ice-cream travels on the first carousel at a
constant velocity of 1 until it reaches position $\alpha$. As soon as
the ice-cream reaches $\alpha$, signal $S1$ is emitted with the TAG
value \texttt{Storage1}, say. Signal $S2$ is emitted instantaneously and
the hybrid automaton moves to mode \textbf{B}. In this mode, the
ice-cream and the diverter, both move at a constant velocity until the
diverter covers the distance of $\theta$ arc-length units. Finally, a
transition is made to mode \textbf{D}, where any further distance until
$\beta$ is covered by the ice-cream and then the ice-cream moves onto
the second carousel and is placed into the correct storage.

The movement of the ice-cream for this hybrid automaton assuming
$\alpha=3$, $\beta=10$ and $\theta=6$ is shown in
Figure~\ref{fig:1b}. Assuming instantaneous discrete mode-switch model
of the hybrid automaton, choosing $\alpha=3$ is a feasible solution as
seen in Figure~\ref{fig:1b}. The ice-cream is detected at position 3 on
the first carousel and an \textit{instantaneous} move is made to control
mode \textbf{B} where the ice-cream moves another 6 units ending up at
position 9 when the hybrid automaton is in mode \textbf{D}, which is
less than $\beta = 10$.

\subsection{The synchronous controller}
\label{sec:synchr-model-comp}

\begin{definition}
 A synchronous controller is a tuple $(Q,q_0,I,O,A,T)$ where:
 \begin{compactitem}
 \item $Q$ is the set of states
 \item $q_0 \in Q$ is the starting state
 \item $I$ is the set of input signals
 \item $O$ is the set of output signals
 \item $A$ is the set of actions
 \item $T$ is the transition relation: $T \subseteq Q \times
   \mathcal{B}(I) \times 2^A \times 2^O \times Q$. $\mathcal{B}(I)$ is a
   Boolean expression over the symbols in $I$.
 \end{compactitem}
\end{definition}

Simply put, a synchronous controller is a directed graph with edges
carrying the labels of the form $b/A',O': b \in \mathcal{B}(I), A'
\subseteq A, O' \subseteq O$. Intuitively, each edge can be taken if the
Boolean condition on the edge holds true. Furthermore, actions
(functions) are performed and output signals emitted upon taking the
transition.  

\subsubsection{The timing semantics of synchronous controllers}
\label{sec:timing-semant-synchr}

\begin{figure}[t!]
  \centering 
  \subfloat[Synchronous controller controlling a plant]
  {\includegraphics[scale=0.4]{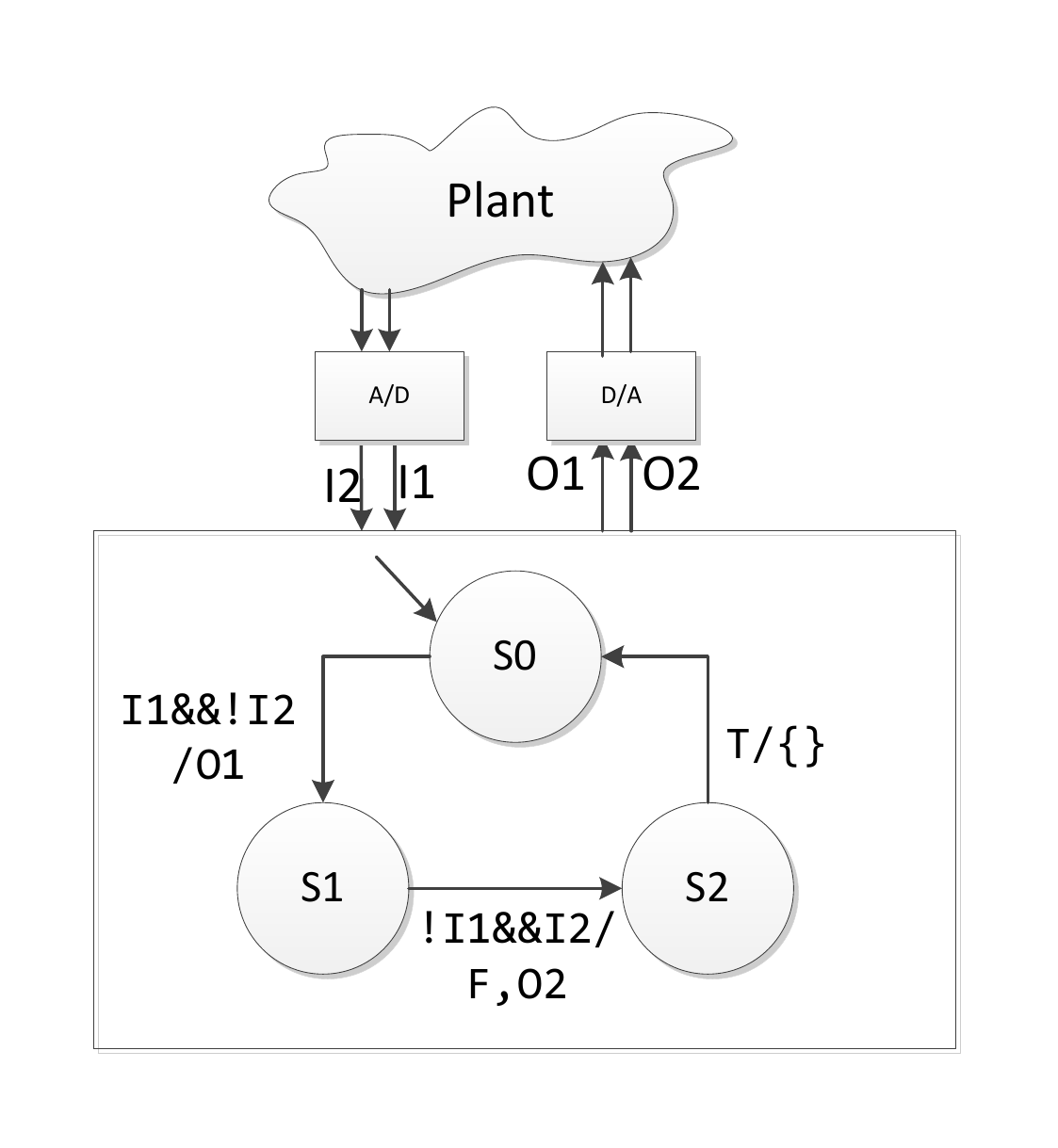} \label{fig:4a}}
  \subfloat[Synchronous controller timing
  diagram]{\includegraphics[scale=0.5]{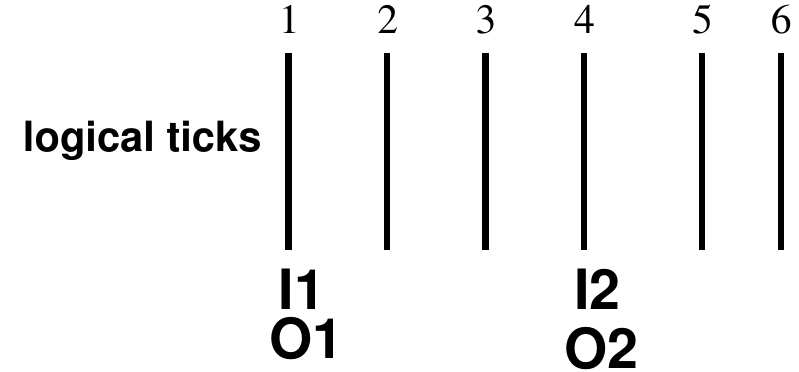}\label{fig:4b}}
  \caption{An example synchronous controller and its timing diagram}
  \label{fig:4}
\end{figure}

Figure~\ref{fig:4a} shows a simple example of a synchronous controller
controlling a plant. The controller's input signal set is $\{I1, I2\}$
and output signals are produced from the set $\{O1, O2\}$. The
transition system for the controller is also shown in
Figure~\ref{fig:4a}. There are three states in the transition
system. The initial state is labeled \textbf{S0}. When signal $I1$ is
produced from the plant, the controller makes a transition to state
\textbf{S1}. In the process also emitting signal $O1$ back to the
plant. Next, when signal $I2$ is produced from the plant, the controller
makes a transition to state \textbf{S2}. Furthermore, the controller
performs an action $F$ and outputs signal $O2$ back to the plant.

The timing diagram for the controller is shown in
Figure~\ref{fig:4b}. Every synchronous controller, following the zero
delay model~\cite{berry96}, progresses in lockstep with a logical clock
tick. The inputs are captured from the plant at the start of the logical
tick, a reaction function is called to process these inputs (in this
case the reaction function is the transition system in
Figure~\ref{fig:4a}) and finally the outputs are produced at the end of
the tick. The logical ticks are shown as bars in Figure~\ref{fig:4b}. At
logical tick 1, the input signal $I1$ is captured from the plant, and
the output signal $O1$ is instantaneously produced at the end of the
logical tick. Similarly, input signal $I2$ is captured at the start of
tick 4 and output signal is emitted back to the plant at the end this
tick -- \textit{instantaneously}.

Unfortunately, execution of every reaction to the input signals takes
some $\delta$ physical time. The zero delay model implicitly requires
that the reaction to the input signals be fast enough in order to not
miss any input events from the plant. In order to satisfy this implicit
restriction, we need to calculate the \textit{Worst Case Reaction Time}
(WCRT) from amongst all the reaction times, which needs to be shorter
than the inter-arrival between any two incoming events. Formally, let
$\{\delta_1,\ldots,\delta_N\}$ be the set of all possible reaction times
for some synchronous controller. Then, $\exists i \in N$, where
\mbox{$WCRT = max (\delta_i)$}. WCRT of any synchronous controller can
be calculated statically irrespective of the plant model. Many different
techniques exist for the calculation of the WCRT of a synchronous
controller~\cite{wilhelm08}.


\section{The problem of time-delayed mode switches}
\label{sec:motivating-example}

  

Let us revisit the manufacturing control system example in
Section~\ref{sec:an-example-linear-1} and use a synchronous language to
implement the \texttt{TRDC} controller that performs the discrete mode
switches in Figure~\ref{fig:1b}. Since the length ($\beta$), the width
($\theta$) of the first carousel and the speed of movement of the
carousel and the diverter are all fixed, we only need to place the
\texttt{TRDC} at the correct position on the first carousel so that the
diverter is in the correct position by the time ice-cream reaches
position $\beta$. Our \textit{goal} is to statically verify that any
ice-cream on the first carousel will be diverted to the correct storage
depending upon its tag. A hybrid automaton should help us model this
system to guarantee this \textit{safety} property. Note that the reader
should interpret the term \textit{verify} loosely, because the
reachability problem for hybrid automata are known
undecidable~\cite{alur1993hybrid}.

\subsection{The hybrid automaton and the worst case reaction time of
  synchronous controllers}
\label{sec:hybrid-model-worst}

\begin{figure}[t!]
  \centering \subfloat[Difference between the modeled system and the
  real system. ]{\scalebox{0.5}{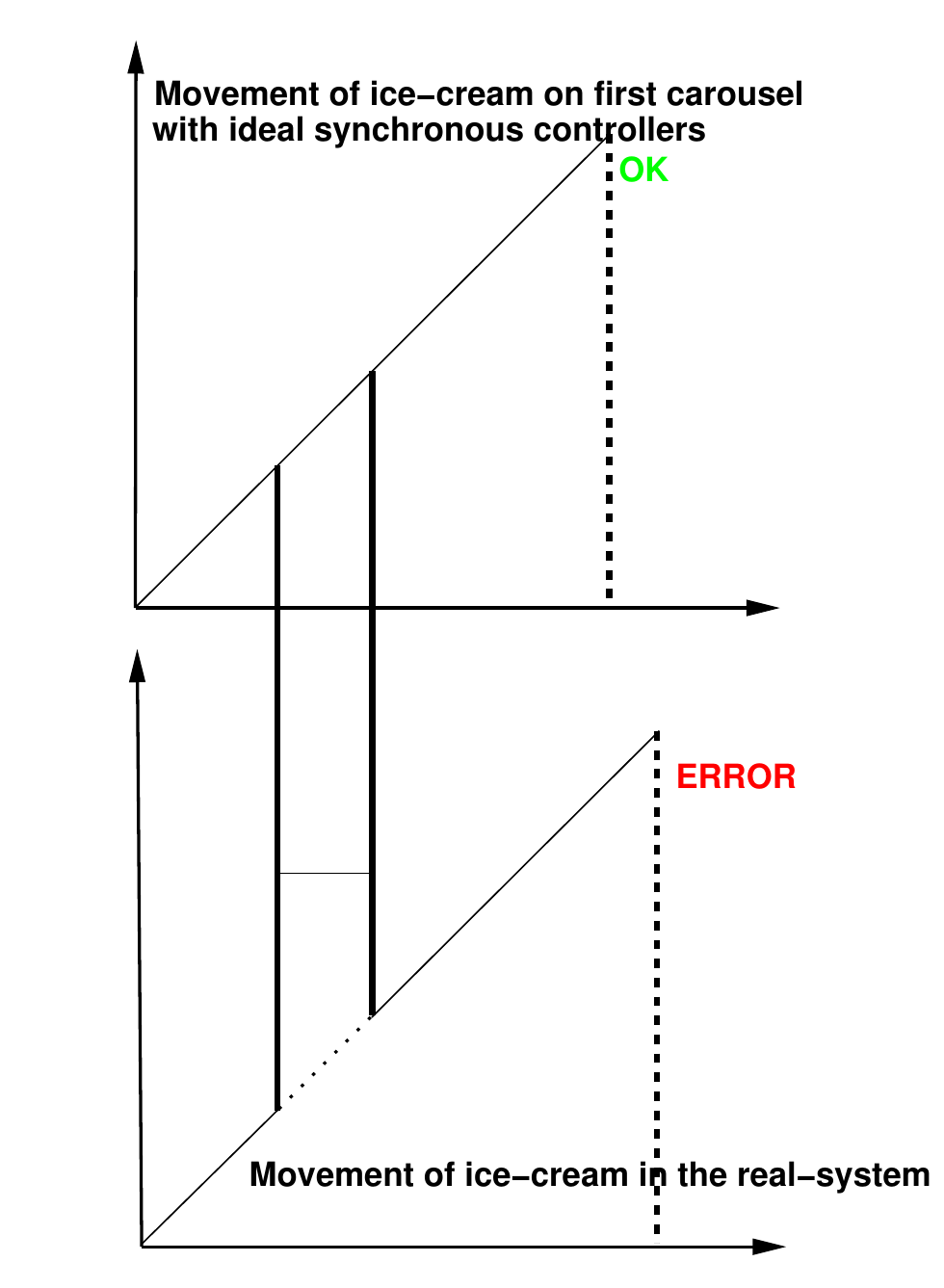} \label{fig:2a}}
  \subfloat[Missed item tag due to WCRT]
  {\scalebox{0.5}{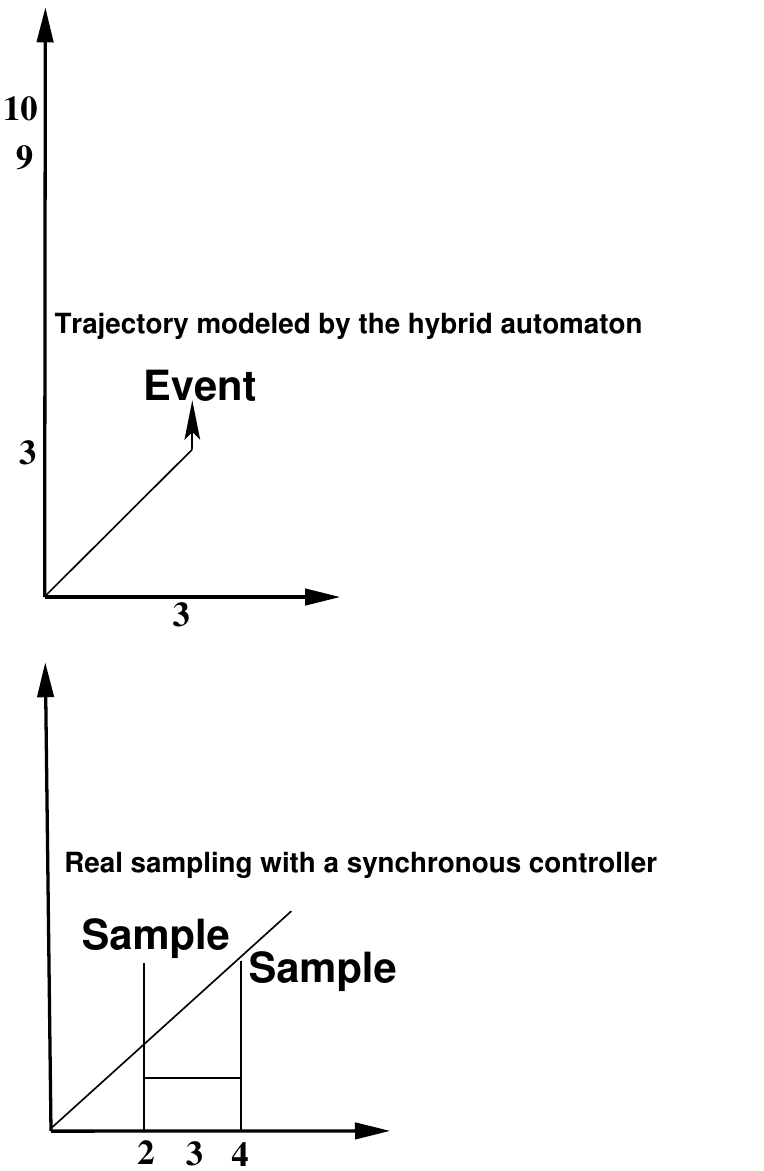}\label{fig:2b}}
  \caption{The different movement of the ice-cream -- hybrid automaton
    model vs. the real system}
  \label{fig:2}
\end{figure}

The movement of the ice-cream in the real system with the \texttt{TRDC}
placed at position 3 (as obtained from the hybrid automaton model) is
shown in Figure~\ref{fig:2a}, bottom graph. Every decision made by the
controller does take some time. In case of synchronous controllers, this
time is the WCRT. Suppose that $WCRT = 2$ units for the \texttt{TRDC}
controller, then the ice-cream is at position 5 when the hybrid
automaton moves to mode \textbf{B}. Now, the system modeled by the
hybrid automaton and the real implementation are not in-sync. In fact,
when the system enters mode \textbf{D}, the invariant $x \leq \beta$
does not hold and the transition is immediately made back to mode
\textbf{A}. But, the ice-cream is already at position 11 when the system
enters mode \textbf{D}, which is past $\beta = 10$ and hence, the
ice-cream now moves to \texttt{Storage2} rather than \texttt{Storage1}
as desired, thereby violating the safety property. Overall, the model
does \textit{not} reflect reality and needs to be modified. One might
\textit{assume} that the transition time of the controller is orders of
magnitude smaller compared to the speed of movement of the ice-cream on
the carousel and hence, can be considered as zero. This is a very rough
\textit{approximation} as indicated
in~\cite{albert2004comparison}. There are data acquisition delays,
sensor delays, communication delays, computation delays in digital
controllers, which cannot be ignored with the slight of hand. These
delays need to be accounted for in the WCRT of the embedded controller.


\subsection{The hybrid automaton, the worst case reaction time and the
  synchrony hypothesis}
\label{sec:hybrid-model-worst-1}

Every synchronous program can be statically analyzed to find its
WCRT. As mentioned before (see Section~\ref{sec:synchr-model-comp}) the
synchrony hypothesis is guaranteed \textrm{iff} the inter-arrival time
of input events is less than or equal to $WCRT$. For the manufacturing
system example $WCRT = 2$, hence, the synchronous control logic
(\texttt{TRDC}), \textit{in the worst case}, samples inputs every 2
units of time. An input is generated when the ice cream reaches position
3 (since $\alpha = 3$), but under the synchrony assumption, this input
is missed as this input event is not aligned with the edge of the
controller clock, i.e., it is not divisible by $WCRT=2$ (see
Figure~\ref{fig:2b}). An event driven system would, on the other hand,
easily capture this input event. Hence, there is an implicit assumption
in the hybrid automaton that the control logic is event driven rather
than clock-driven as is the case with synchronous controllers. This is
yet another problem that needs to be addressed when designing
synchronous controllers.

The aforementioned problems occur due to the non-zero reaction time of
the synchronous controllers. More precisely, the plant makes progress
while the controller carries out internal computations, unlike in the
hybrid automaton where these discrete mode-switches zero time. This
plant behavior could be modeled by labeling the discrete transitions in
the hybrid automaton with differential equations. But, this solution
does not bode well with the semantics of the hybrid automaton. The time
for the discrete transition depends upon the implementation of the
controller, which differs depending upon the underlying platform,
compiler technology, etc. Hence, if we were to simply label the discrete
transitions with differential equations, the evolution of the continuous
plant variables would depend upon the speed of the controller, which is
in stark contrast to the semantics of the hybrid
automaton~\cite{Henzinger:1996:THA:788018.788803}. In light of these
problems we need a new programming model for design and verification of
hybrid systems. In the rest of the paper we present a power language
called \textit{HySysJ} that: (1) results in high fidelity hybrid system
models, by incorporating time-delayed mode switching, (2) allows
automatically extracting controllers for implementation from the hybrid
model and (3) allows for automatic formal verification of the hybrid
system.



\section{The base language}
\label{sec:base-language}

The proposed language HySysJ builds atop the synchronous subset of the
SystemJ~\cite{amal10} programming language, which is itself inspired
from Esterel~\cite{berry96}. The core kernel statements of the language
are given in Figure~\ref{fig:5}.

\begin{figure}[t!]
  \centering
   \[
   \begin{array}{ll}
     stmt & ::=  stmt_0 [\mathrm{``||"} stmt] \\
     stmt_0 & ::= stmt_1 [``;" stmt_0] \\
     stmt_1 & ::= \\
     & | ``\mathbf{nothing}" | ``\mathbf{emit}"\ a | ``?"a ``=" expr \\
     & | ``\mathbf{pause}" | ``\mathbf{abort}" ``("
     [``\mathbf{immediate}"] expr ``)" stmt \\
     & | ``\mathbf{if}" ``(" expr ``)" stmt ``\mathbf{else}"
     stmt \\
     & | ``\mathbf{suspend}" ``("
     [``\mathbf{immediate}"] expr ``)" stmt \\
     & | [``\mathbf{input}" | ``\mathbf{output}" ] [type] ``\mathbf{signal}" a\ [op\ ``="
     expr] \\
     & | ``\mathbf{loop}" stmt
     | ``\{" stmt ``\}" \\
     op & ::= ``\mathbf{op+}" | ``\mathbf{op*}"
   \end{array}
   \]
   \caption{The core kernel statements of the synchronous subset of
     SystemJ. The terminals appear within double quotes, and angular
     brackets indicate optional syntactic components.}
  \label{fig:5}
\end{figure}

The core synchronous language constructs in SystemJ are borrowed
directly from Esterel. The $\mathbf{nothing}$ construct terminates
instantaneously and is primarily used in the structural operational
semantics during term rewriting. Every signal is declared via the
$\mathbf{signal}$ declaration statement. The $type$ declaration for a
signal is optional. A non-typed signal is considered to be a pure signal
whose \texttt{status} can be set to $true$ for one logical tick by
emitting it (via $\mathbf{emit}$) and is $false$ if it is not emitted in
that logical tick. A \textit{valued} signal has a value and a
status. Every valued signal is uniquely associated with one of the
types: $\mathbf{ratio}$, $\mathbf{integer}$, or $\mathbf{boolean}$. A
signal can be emitted multiple times with different values in the same
logical tick. In such cases, signal values are combined with operators
defined during signal declaration. Only associative and commutative
operators (e.g., $\mathbf{op+}$ and $\mathbf{op*}$) are permitted over
signal values and everything must be well-typed in the expected
way. Unlike the status of a signal, the value of a signal is persistent
over logical ticks. A block of statements can be preempted or suspended
for a single tick using $\mathbf{abort}$ and $\mathbf{suspend}$
constructs, respectively. The $\mathbf{if}$ construct is the usual
branching construct, operating on the status or values of
signals. Moreover, one or more of the aforementioned statements can be
run in lockstep parallel using the synchronous parallel operator $||$.
Finally, the $\mathbf{loop}$ construct is used to write temporal loops,
whereby each iteration consumes a logical tick via the $\mathbf{pause}$
construct.

The synchronous semantics of SystemJ \textit{differs} from Esterel in
one significant way: the emission of every signal is \textit{delayed by
  a single logical tick} and is only visible in the next iteration of
the synchronous program. We describe these so called \textit{delayed
  signal semantics} using simple code snippets shown in
Figure~\ref{fig:8}.

\newbox{\csf}
\begin{lrbox}{\csf}
  \begin{lstlisting}[style=sysj,morekeywords={signal,loop,abort,await,emit,present,trap,pause,exit,delay,suspend}]
    signal S, A, B;
    emit S; 
    if (S) emit A else emit B;
    pause
  \end{lstlisting}
\end{lrbox}

\newbox{\css}
\begin{lrbox}{\css}
  \begin{lstlisting}[style=sysj,morekeywords={signal,loop,abort,await,emit,present,trap,pause,exit,delay,suspend}]
    signal S, A, B;
    emit S; 
    pause; //additional pause 
    if (S) emit A else emit B;
    pause
  \end{lstlisting}
\end{lrbox}

\begin{figure}[t!]
  \centering
  \subfloat[Synchronous program code snippet 1] {
    \usebox\csf
    \label{fig:8a}
  }
  \hspace{30pt}
  \subfloat[SystemJ logical timing behavior] {
    \includegraphics[scale=0.5]{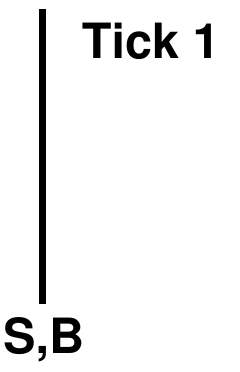}
    \label{fig:8c}
  }

  \subfloat[Synchronous program code snippet 2] {
    \usebox\css
    \label{fig:8d}
  }
  \hspace{30pt}
  \subfloat[SystemJ logical timing behavior] {
    \includegraphics[scale=0.5]{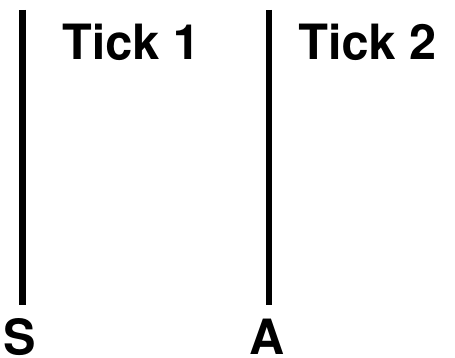}
    \label{fig:8f}
  }
  \caption{SystemJ vs. Esterel logical timing behavior}
  \label{fig:8}
\end{figure}

Figure~\ref{fig:8a} shows a very simple synchronous program. Three pure
signals \texttt{S}, \texttt{A}, and \texttt{B} are declared. Signal
\texttt{S} is emitted and then its status is checked for presence, if
this signal has been emitted, then signal \texttt{A} is emitted, else
signal \texttt{B} is emitted. Finally, the program ends with the
$\mathbf{pause}$ statement indicating the end of the logical tick. The
logical timing behavior of this SystemJ program is shown in
Figure~\ref{fig:8c}. In SystemJ, the emission of signal makes its
visible only in the next logical tick, hence, this program emits signal
\texttt{B} in the first logical tick. The logical timing behavior
achieved by slightly changing the program (inserting an additional
$\mathbf{pause}$ construct) is shown in Figures~\ref{fig:8d}
and~\ref{fig:8f}. In this case, since signal \texttt{S} is emitted in
tick-1, its status is \texttt{true} in tick-2 and hence, SystemJ
following the so called \textit{delayed} signal semantics emits signal
\texttt{A} in the second tick.

Valued signals follow rules similar to signal statuses, i.e., reading a
value of the signal (e.g., $?S$) always gives the value from the
previous logical tick or the default value (0 usually), while setting
the value of the signal (e.g., $?S = 2$) always sets the current
value. The previous value of the signal is updated to the current value
at the end of the logical tick.

\newbox{\causalf}
\begin{lrbox}{\causalf}
  \begin{lstlisting}[style=sysj,morekeywords={signal,loop,abort,await,emit,present,trap,pause,exit,delay,suspend}]
    if (S) else emit S
  \end{lstlisting}
\end{lrbox}

\newbox{\causals}
\begin{lrbox}{\causals}
  \begin{lstlisting}[style=sysj,morekeywords={signal,loop,abort,await,emit,present,trap,pause,exit,delay,suspend}]
    ?S = ?S + 1
  \end{lstlisting}
\end{lrbox}

\begin{figure}[t!]
  \centering
    \usebox\causals
    \caption{Code snippet -- incorrect in Esterel, but correct in
      SystemJ}
  \label{fig:11}
\end{figure}

This so called delayed signal semantics implicitly avoid plethora of
problems that plague Esterel programs, related to causality. 
Consider the code snippet in Figure~\ref{fig:11}. In case of Esterel,
the value of signal \texttt{S} is fed back to itself in the same logical
tick, and hence, in Esterel, one needs to check that \mbox{\texttt{?S ==
    (?S + 1)}}, which obviously has no solution. But, in SystemJ, since
the signal values are only ever updated at the end of a logical tick,
the program in Figure~\ref{fig:11} is computable.

\newbox{\flisting}
\begin{lrbox}{\flisting}
  \begin{lstlisting}[style=sysj,morekeywords={signal,loop,abort,await,emit,present,trap,pause,exit,delay,suspend}]
    int signal S op+ = 0;
    signal R;
    abort (R) 
     loop {
       ?S = ?S + 1;
       if (?S == (d-1)) emit R;
       pause;
     }
  \end{lstlisting}
\end{lrbox}

Now that we have described the base language and its syntactic
constructs, we are ready to introduce the continuous elements into the
synchronous subset of SystemJ that will result in the new HySysJ hybrid
system specification language.

\section{HysysJ -- introducing continuous time in synchronous SystemJ}
\label{sec:example-driv-inform}

The most fundamental modification to the synchronous language described
in Section~\ref{sec:base-language} is the introduction of continuous
variables and related actions that manipulate these variables. Following
standard practice, we will first introduce the syntactic extensions and
then describe the semantics.

\newbox{\rgone}
\begin{lrbox}{\rgone}
  \begin{lstlisting}[mathescape,style=sysj,morekeywords={signal,loop,abort,await,emit,present,trap,pause,exit,delay,suspend},escapechar=|]
    signal R;
    abort (R)|\label{lst:12}| 
     loop {|\label{lst:1}|
       a = a + $\rho$ * WCRT; |\label{lst:2}|
       if (!TTL ($[a'=\rho]$, $expr$, $\{a\}$)) emit R; |\label{lst:3}|
       pause|\label{lst:4}|
     }|\label{lst:5}|
  \end{lstlisting}
\end{lrbox}

\newbox{\rgtwo}
\begin{lrbox}{\rgtwo}
  \begin{lstlisting}[mathescape,style=sysj,morekeywords={signal,loop,abort,await,emit,present,trap,pause,exit,delay,suspend},escapechar=|]
    signal R;
    abort (R) 
     loop {
       a = a + $\rho$ * WCRT; b = b + $\sigma$ * WCRT; |\label{lst:13}|
       if (!TTL ($[a'=\rho, b'=\sigma]$, $expr$, $\{a, b\}$)) emit R; 
       pause
     }
  \end{lstlisting}
\end{lrbox}


\begin{figure}[t!]
  \centering
  \subfloat[The syntactic constructs for continuous variable declaration
  and manipulation] {
    \begin{small}
      
    \begin{tabular}{ll}
      $stmt_1$ & $::= stmt_2 $\\
      $stmt_2$ & $::= $ \\
      & $| ``\mathbf{cont}" a\ [op\ ``="\ expr]$ \\
      & $| ``\mathbf{cont}" a\ [``="\ expr]$ \\
      & $|\ a = expr $\\
      & $| ``\mathbf{do} ``\{" stmt_3 ``\}" ``\mathbf{until}" ``(" expr ``)" $\\
      $stmt_3 $ & $::=  $\\
      & $|\ a``'" = expr $\\
      & $|\ a``'" = expr [``||" stmt_3]$ 
    \end{tabular}
    \end{small}
     \label{fig:7a}
  }
  
  \subfloat[The rewrite for the derived construct:
  \mbox{$\mathbf{do}\ \{a' = \rho\} \mathbf{until} (expr)\}$}.] {
    \usebox\rgone \label{fig:7c}}
  \hspace{30pt}
  \subfloat[The rewrite for the derived construct:
  \mbox{$\mathbf{do}\ \{a' = \rho || b' = \sigma \} \mathbf{until}
    (expr)\}$}.] {
    \usebox\rgtwo \label{fig:7b}}
  \caption{The continuous variables and derived construct operating on
    these variables in HySysJ}
  \label{fig:7}
\end{figure}

\subsection{Syntax of continuous actions}
\label{sec:synt-cont-acti}

The syntactic extensions to declare and manipulate the continuous
variables are given in Figure~\ref{fig:7a}. Every continuous variable is
declared with the qualifier $\mathbf{cont}$. A default value can be
specified during declaration. Uninitialized continuous variables take a
default value of 0. Furthermore, a commutative and associative operator
($op$) can be used to combine the values of the continuous variables,
just like in case of valued signals. The type of every continuous
variable is a $\mathbf{ratio}$.

Two forms of syntactic extensions are allowed for manipulating
continuous variables: (1) a direct assignment to the continuous variable
or use of continuous variables in expressions, called
\textit{instantaneous actions} and (2) writing first order ODEs inside a
$\mathbf{do}\ \mathbf{until}\ (expr)$ block that evolve the continuous
variables, called \textit{flow actions}. We use primed symbols (e.g.,
\mbox{\texttt{a$'$ = c}}, where \texttt{c} is some constant) to describe
these first order derivatives. One or more such ODEs can be specified
inside the $\mathbf{do}$ block. The synchronous parallel operator $||$
is used to specify more than one ODE inside the $\mathbf{do}$
block. Every ODE inside the $\mathbf{do}$ block is evaluated
\textit{simultaneously} until the $expr$ (the so called invariant
condition) holds true. The $\mathbf{until}$ $expr$ is required to
evaluate to a Boolean $true$ or $false$ value.

\subsection{Semantics of continuous actions}
\label{sec:semant-cont-acti}


\subsubsection{Instantaneous actions}
\label{sec:inst-acti}

Instantaneous actions are so called, because the statement terminates
instantaneously without consuming a logical tick. Examples of
instantaneous actions are shown in Figure~\ref{fig:15}. These include;
assigning a value to a continuous variable, reading the value of a
continuous variable, assigning the value of the continuous variable to a
valued signal or another continuous variable, etc.

Continuous variables, like signals, follow delayed semantics. Hence,
using a continuous variable in an expression (right hand side in case of
an assignment statement) always gives the value from the previous
logical tick or the default value. A new value is assigned to a
continuous variable only at the end of the current logical tick. 

In Figure~\ref{fig:16a}, continuous variable \texttt{a} is first
declared, with a default value of 0, and then assigned a value of
1. Next, an \texttt{if} \texttt{else} block is used to check the value
of \texttt{a}. If the value of \texttt{a} is 1, then signal \texttt{S1}
is emitted else signal \texttt{S2} is emitted. The same program is
presented in Figure~\ref{fig:16b}, except that a \texttt{pause}
statement is inserted after the assignment statement: \texttt{a = 1}.
In the first case, due to delayed semantics, when the value of
\texttt{a} is read in the \texttt{if} expression, the return value is 0
(the default value) and hence, signal \texttt{S2} is emitted. On the
other hand, in Figure~\ref{fig:16b} when the program flow reaches the
\texttt{if} statement, it is the second logical tick and hence, the
value of \texttt{a} is 1 (assigned in the previous logical tick) thus
signal \texttt{S1} is emitted.

It is important to note that the name \textit{instantaneous action} does
not mean that the value of the continuous variable changes
\textit{instantaneously}. Every continuous variable changes its value
only at the end of the tick. The name instantaneous action \textit{only}
implies that the statement itself is instantaneous and does not consume
logical ticks\footnote{Every statement, except for \texttt{pause} in
  HySysJ is instantaneous}.

\newbox{\clfggg}
\begin{lrbox}{\clfggg}
  \begin{lstlisting}[mathescape,style=sysj,morekeywords={until,cont,signal,loop,abort,await,emit,present,trap,pause,exit,delay,suspend}]
    signal S1, S2; //declaring pure signals
    cont a; //declaring a with default value 0
    a = 1; // assigning value 1 to continuous variable a
    pause;
    if (a == 1) emit S1 //continuous variable used in expression.
    else emit S2;
    pause
  \end{lstlisting}
\end{lrbox}

\newbox{\clfgg}
\begin{lrbox}{\clfgg}
  \begin{lstlisting}[mathescape,style=sysj,morekeywords={until,cont,signal,loop,abort,await,emit,present,trap,pause,exit,delay,suspend}]
    signal S1, S2; //declaring pure signals
    cont a; //declaring a with default value 0
    a = 1; // assigning value 1 to continuous variable a
    if (a == 1) emit S1 //continuous variable used in expression.
    else emit S2;
    pause
  \end{lstlisting}
\end{lrbox}

\newbox{\clfg}
\begin{lrbox}{\clfg}
  \begin{lstlisting}[mathescape,style=sysj,morekeywords={until,cont,signal,loop,abort,await,emit,present,trap,pause,exit,delay,suspend}]
    cont a = 0; //declaring a continuous variable with initial value 0
    a = 1; // assigning value 1 to continuous variable a
    if (a == 1) emit S; //continuous variable used in expression.
    ?S = a; //value of continuous variable a assigned to a valued signal.
  \end{lstlisting}
\end{lrbox}

\begin{figure}[t!]
  \centering
  \usebox\clfg
  \caption{Instantaneous actions on continuous variables in HySysJ}
  \label{fig:15}
\end{figure}

\begin{figure}[t!]
  \centering
  \subfloat[Code snippet 1]{\usebox\clfgg \label{fig:16a}}
  
  \subfloat[Code snippet 2]{\usebox\clfggg \label{fig:16b}}
  \caption{Instantaneous actions on continuous variables in HySysJ with
    delayed semantics}
  \label{fig:16}
\end{figure}

\subsubsection{Flow actions}
\label{sec:flow-actions}

The flow actions are programmed using $\mathbf{do}$ $\mathbf{until}$
blocks and are first order ODEs with a constant rate of change. In the
example in Figure~\ref{fig:12a}, continuous variable \texttt{a} is
declared and initialized to a value of 0, which is an instantaneous
action. Next, this variable evolves continuously until its value is 2
inside a $\mathbf{do}$ $\mathbf{until}$ block. In the next example, two
variables; \texttt{a} and \texttt{b} evolve together until the
\textit{invariant condition} ($\mathbf{until}$ expression) holds. One
can also combine multiple such flow actions together in synchronous
parallel (Figure~\ref{fig:12d}). 
Finally, HySysJ also allows preempting flow actions using the standard
preemptive constructs from the base language.

\newbox{\clf}
\begin{lrbox}{\clf}
  \begin{lstlisting}[mathescape,style=sysj,morekeywords={until,cont,signal,loop,abort,await,emit,present,trap,pause,exit,delay,suspend}]
    cont a = 0;
    do {a$'$ = 1} until (a <= 2) 
  \end{lstlisting}
\end{lrbox}

\newbox{\cls}
\begin{lrbox}{\cls}
  \begin{lstlisting}[mathescape,style=sysj,morekeywords={until,cont,signal,loop,abort,await,emit,present,trap,pause,exit,delay,suspend}]
    cont a = 0, b = 0; 
    do {a$'$ = 2 || b$'$ = 2} until (a <= 16 && b <= 10)
  \end{lstlisting}
\end{lrbox}

\newbox{\clfi}
\begin{lrbox}{\clfi}
  \begin{lstlisting}[mathescape,style=sysj,morekeywords={until,cont,signal,loop,abort,await,emit,present,trap,pause,exit,delay,suspend}]
    cont a = 0, b = 0;
    do {a$'$ = 1 || b$'$ = 1} until (a <= 10 && b <= 6)
  \end{lstlisting}
\end{lrbox}

\newbox{\clt}
\begin{lrbox}{\clt}
  \begin{lstlisting}[mathescape,style=sysj,morekeywords={until,cont,signal,loop,abort,await,emit,present,trap,pause,exit,delay,suspend}]
    cont a = 0, b = 0; 
    do {a$'$ = 1} until (a <= 10) || do {b$'$ = 1} until (b <= 6)
  \end{lstlisting}
\end{lrbox}

\newbox{\clff}
\begin{lrbox}{\clff}
  \begin{lstlisting}[mathescape,style=sysj,morekeywords={until,cont,signal,loop,abort,await,emit,present,trap,pause,exit,delay,suspend}]
    signal S;
    cont a = 0; 
    abort(S) { do {a$'$ = 1} until (true) } || {pause; emit S; pause}
  \end{lstlisting}
\end{lrbox}

\begin{figure}[t!]
  \centering
  \subfloat[Example with one continuous variable]{\usebox\clf \label{fig:12a}}
  
  \subfloat[Example with two continuous variables]{\usebox\cls \label{fig:12b}}

  \subfloat[Another example of two continuous variables]{\usebox\clfi \label{fig:12c}}

  \subfloat[Example of parallel composition of flow actions]{\usebox\clt \label{fig:12d}}

  \subfloat[Example of preemption of continuous variable
  evolution]{\usebox\clff \label{fig:12e}}
  \caption{Examples of continuous actions in HySysJ}
  \label{fig:12}
\end{figure}

\paragraph{\textit{Semantics and intuitive explanation for simple flow actions}}
\label{sec:flow-actions-1}

In this section we describe the rewrite semantics of simple flow actions
and give the intuitive explanation for these rewrites. A complete formal
treatment is provided in Appendix~\ref{sec:disc-rewr-semant}. 

Consider the simple flow action in Figure~\ref{fig:12a}; variable
\texttt{a} evolves linearly with time until it reaches the value 2. The
first order ODE in Figure~\ref{fig:12a} has the solution:
\mbox{$a = \int 1 \times \mathrm{d}t = 1 \times t + C$}, where 1 is the
rate of change of \texttt{a} and $C$ is the initial value of
\texttt{a}. Furthermore, the $\mathbf{until}$ expression gives the upper
bound on this indefinite integral. For this very simple flow action, the
upper bound of the indefinite integral is 2. Hence, the value of
\texttt{a} is: $a = [t]^2_0 + C = 2 + C$. From Figure~\ref{fig:12a}, we
also know that the initial value of \texttt{a} is 0, i.e.,
$ C = 0, \therefore a = 2$.

\begin{eqnarray}
  \label{eq:7}
  \int_0^k \rho \times \mathrm{d}t + C = \sum_{n=0}^{k-1}\ 
  \rho \times \Delta t + a[0] \\
  \label{eq:8}
  \therefore
  a[k] = a[0] + \sum_{n=0}^{k-1}\ \rho \times WCRT
\end{eqnarray}

The \textbf{main idea} of our rewrite is to approximate the continuous
evolution of \texttt{a} using a discrete time
model. Equation~(\ref{eq:7}) gives this approximation. We take advantage
of the synchronous nature of our programming language. Every HySysJ
program proceeds in discrete logical ticks, the value of variable
\texttt{a} at tick 0 (the initial value) is denoted $a[0]$. Similarly,
for some tick $n$, the value is denoted by $a[n]$. In
Equation~(\ref{eq:7}), $\Delta t$ is the time between two discrete
logical ticks, which is the WCRT as stated in
Section~\ref{sec:synchr-model-comp} and can be computed statically for
any HySysJ program~\cite{wilhelm08}. $\rho$ is the rate of change, which
is always a constant for any linear ODE. Finally, the upper bound of the
summation ($k-1$) is dependent upon the $\mathbf{until}$ expression.

Equation~(\ref{eq:8}) obtained from Equation~(\ref{eq:7}) is clearly a
bounded reduction on \texttt{a} (using sum), which in any imperative
language is written using a bounded loop. Hence, our rewrite for any
linear ODE is a bounded temporal loop computing the value of the
continuous variable as shown in Figure~\ref{fig:7c}. In
Figure~\ref{fig:7c} the temporal loop performing reduction on \texttt{a}
spans from lines~\ref{lst:1} to~\ref{lst:5}. The actual reduction is
performed on line~\ref{lst:2}. This temporal loop is exited using a
combination of $\mathbf{emit}$ and $\mathbf{abort}$ constructs as shown
on lines~\ref{lst:12} and~\ref{lst:7}. Finally, the upper bound $k$ is
Equation~(\ref{eq:8}) is computed dynamically in the rewrite using
procedure \texttt{TTL} (\textit{Time To Live}) shown in
Algorithm~\ref{alg:1}. In the general case it is impossible to
statically (at compile time) compute the upper bound $k$ in
Equation~(\ref{eq:8}), since one can have complex invariant conditions
specified in the $until$ expressions and hence, dynamically (at program
execution time) deciding when to abort the temporal loop is the only
viable option.

Delayed semantics play a crucial role in the rewrite of
Figure~\ref{fig:7c}.
\begin{compactitem}
\item Computability of the reduction: Reading the value of \texttt{a}
  (line~\ref{lst:3}) always gives the value from the previous
  tick. Writing to \texttt{a} succeeds only at the end of the logical
  tick, i.e., when the control flow reaches the $\mathbf{pause}$
  construct on line~\ref{lst:4}. The delayed semantics make the
  reduction computable. Moreover, the updated value of \texttt{a} is
  stable and observable only at the end of the tick following delayed
  semantics.
\item The \texttt{TTL} algorithm: The \texttt{TTL} algorithm, which
  decides when to abort the infinitely running temporal loop is also
  dependent upon the delayed semantics. The $\mathbf{abort}$ construct
  (line~\ref{lst:12}) checks if the status of signal \texttt{R} is set
  to \texttt{true} in the \textit{previous} logical tick (statuses of
  signals are \texttt{false} upon declaration), following delayed
  semantics, and if so, aborts the loop performing the reduction. Signal
  \texttt{R} is emitted inside the loop body (line~\ref{lst:3}),
  provided the Boolean value returned from the \texttt{TTL} algorithm in
  \textbf{not} \texttt{true}. The status of \texttt{R} is updated only
  at the end of the tick (line~\ref{lst:4}, which is also completion of
  iteration of the loop). Hence, we are guaranteed that at least one
  iteration of the temporal loop will take place, irrespective of the
  invariant in the $\mathbf{until}$ expression.
  
  From Equation~(\ref{eq:8}), we know that $a[k]$ for some tick $k$
  satisfies the $\mathbf{until}$ invariant. Obviously, $a[k-1]$ should
  also satisfy this invariant condition. But, $a[k+1]$ should never
  satisfy the $\mathbf{until}$ invariant. Due to delayed semantics, we
  now know that given $a[k-1]$, $a[k]$ will always be computed. In order
  not to reach tick $k+1$ (since $a[k+1]$ violates the $\mathbf{until}$
  invariant) signal \texttt{R} should have its status set to
  \texttt{true} at the end of tick $k$. Hence, the signal \texttt{R}
  (line~\ref{lst:3}) should be emitted in the program
  \textit{transition} from tick $k-1$ to $k$ (denoted as
  $[k-1,k)$). But, during this program transition, we only know the
  value $a[k-1]$, which consequently means that algorithm \texttt{TTL}
  needs to look ahead 2 ticks ($k+1-(k-1)=2$) and return a Boolean value
  \texttt{true} if it satisfies the $\mathbf{until}$ invariant and
  \texttt{false} otherwise. If invariant condition is satisfied 2 ticks
  from now, then one more iteration of the loop is allowed to be carried
  out, else the loop terminates at the end of the current program
  transition.
\end{compactitem}

\begin{algorithm}[t!]
  \begin{minipage}{1.0\linewidth}
    \SetKwInOut{Input}{Input}
    \Input{$\Omega$: a list of ODEs from one $\mathbf{do}$ block}
    \Input{$expr$: the $\mathbf{until}$ expression}
    \Input{$\mathcal{V}$: the set of continuous variables in $\Omega$}
    \KwResult{a Boolean value}
    $\Delta$ $\leftarrow$ $\emptyset$\;
    \For {each $v$ in $ \mathcal{V}$} {
      $\tau$ $\leftarrow$ $[\![v]\!]$ + $2$ * get\_rho (filter
      ($\Omega$, $v$)) * WCRT~\footnote{$[\![v]\!]$ is the current
        value of the continuous variable $v$.}\footnote{Function
        get\_rho returns the rate of change for the continuous variable
        $v$.}\label{alg:1:l1}\;
      \tcp{Union the value $\tau$ of $v$ two ticks from now in set $\Delta$}
      $\Delta$ $\leftarrow$ $\Delta \cup \{v \rightarrow \tau \}$\;
    }
    \Return holds\_at\_delta ($expr$, $\Delta$)\;
    \caption{Algorithm to calculate TTL}
    \label{alg:1}
  \end{minipage}
\end{algorithm}

\newbox{\rfirst}
\begin{lrbox}{\rfirst}
  \begin{lstlisting}[mathescape,style=sysj,morekeywords={until,wait_inbetween,cont,signal,loop,abort,await,emit,present,trap,pause,exit,delay,suspend},escapechar=|]
    cont a = 0;
    signal R;
    abort (R) 
     loop {|\label{lst:6}|
       a = a + WCRT;|\label{lst:7}|
       if (!TTL ([$a'=1$], a<=2, $\{a\}$)) |\label{lst:8}|
        emit R;|\label{lst:9}|
       pause|\label{lst:10}|
     }|\label{lst:11}|
  \end{lstlisting}
\end{lrbox}

\newbox{\rsecond}
\begin{lrbox}{\rsecond}
  \begin{lstlisting}[mathescape,style=sysj,morekeywords={until,wait_inbetween,cont,signal,loop,abort,await,emit,present,trap,pause,exit,delay,suspend}]
    cont a = 0, b = 0;
    signal R;
    abort (R) 
    loop {
      a = a + (2 * WCRT);
      b = b + (2 * WCRT);
      if (!TTL ([$a'=2$, $b'=2$], 
         a<=16 && b<=10, $\{a,b\}$)) 
      emit R;
      pause
    }
  \end{lstlisting}
\end{lrbox}

\newbox{\rthird}
\begin{lrbox}{\rthird}
  \begin{lstlisting}[mathescape,style=sysj,morekeywords={until,wait_inbetween,cont,signal,loop,abort,await,emit,present,trap,pause,exit,delay,suspend}]
    cont a = 0, b = 0;
    signal R;
    abort (R) 
    loop {
      b = b + WCRT;
      a = a + WCRT;
      if (!TTL ([$a'=1$, $b'=1$], 
         a<=10 && b<=6, $\{a,b\}$)) 
       emit R;
      pause
    }
  \end{lstlisting}
\end{lrbox}

\newbox{\rfourth}
\begin{lrbox}{\rfourth}
  \begin{lstlisting}[mathescape,style=sysj,morekeywords={until,cont,
      wait_inbetween,signal,loop,abort,await,emit,present,trap,pause,exit,delay,suspend}]
    cont a = 0, b = 0; 
    {
      signal R;
      abort (R) 
       loop {
        a = a + WCRT;
        if (!TTL ([$a'=1$], 
           a<=10, $\{a\}$)) 
         emit R;
        pause;
      }
    } || {
      signal R;
      abort (R) 
       loop {
        b = b + WCRT;
        if (!TTL ([$b'=1$], 
           b<=6, $\{b\}$))
         emit R;
        pause;
      }
    }
  \end{lstlisting}
\end{lrbox}

\newbox{\rfifth}
\begin{lrbox}{\rfifth}
  \begin{lstlisting}[mathescape,style=sysj,morekeywords={until,wait_inbetween,
      cont,signal,loop,abort,await,emit,present,trap,pause,exit,delay,suspend}]
    signal S;
    cont a = 0;
    abort (S) {
      loop {
       a = a+WCRT; 
      loop pause 
     }
    }||{pause; emit S; pause}
  \end{lstlisting}
\end{lrbox}

\begin{figure}[t!]
  \centering
  \subfloat[Rewrite for flow action in Figure~\ref{fig:12a}]{
    \usebox\rfirst
    \label{fig:14a}}
  \hspace{10pt}
  \subfloat[The timing diagram for Figure~\ref{fig:14a}]{
    \includegraphics[scale=0.35]{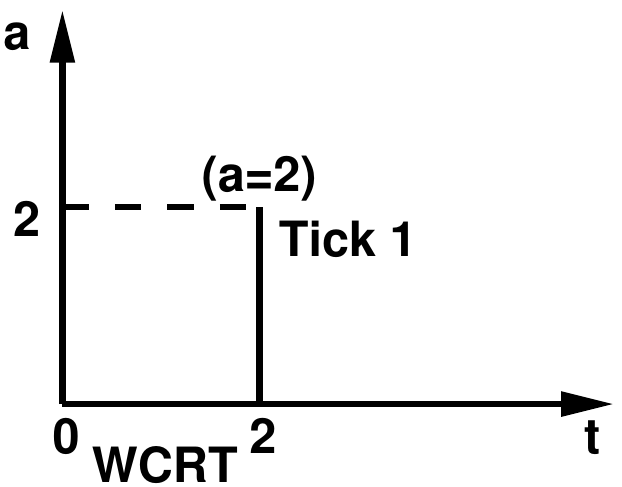}
    \label{fig:14b}}

  \subfloat[Rewrite for flow action in Figure~\ref{fig:12b}]{
    \usebox\rsecond
    \label{fig:14c}}
  \hspace{10pt}
  \subfloat[The timing diagram for Figure~\ref{fig:14b}]{
    \includegraphics[scale=0.35]{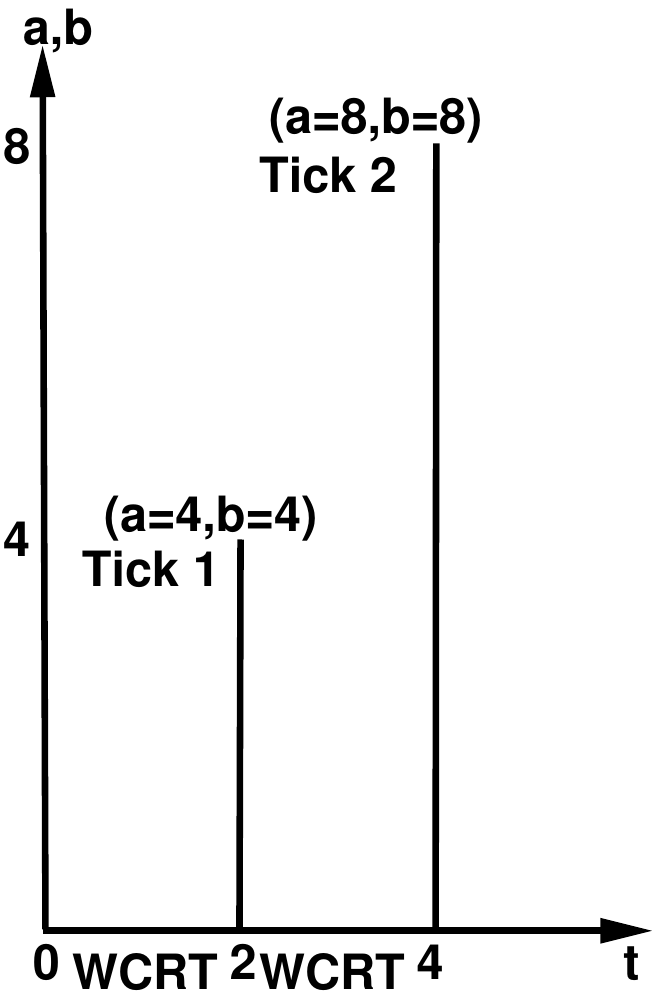}
    \label{fig:14d}}
  
  \caption{Rewrites for the flow actions in Figure~\ref{fig:12}
    with WCRT = 2}
  \label{fig:14}
\end{figure}

\begin{figure}[t!]
  \ContinuedFloat
  \centering
  \subfloat[Rewrite for flow action in Figure~\ref{fig:12c}]{
    \usebox\rthird
    \label{fig:14e}}
  \hspace{10pt}
  \subfloat[The timing diagram for Figure~\ref{fig:14e}]{
    \includegraphics[scale=0.35]{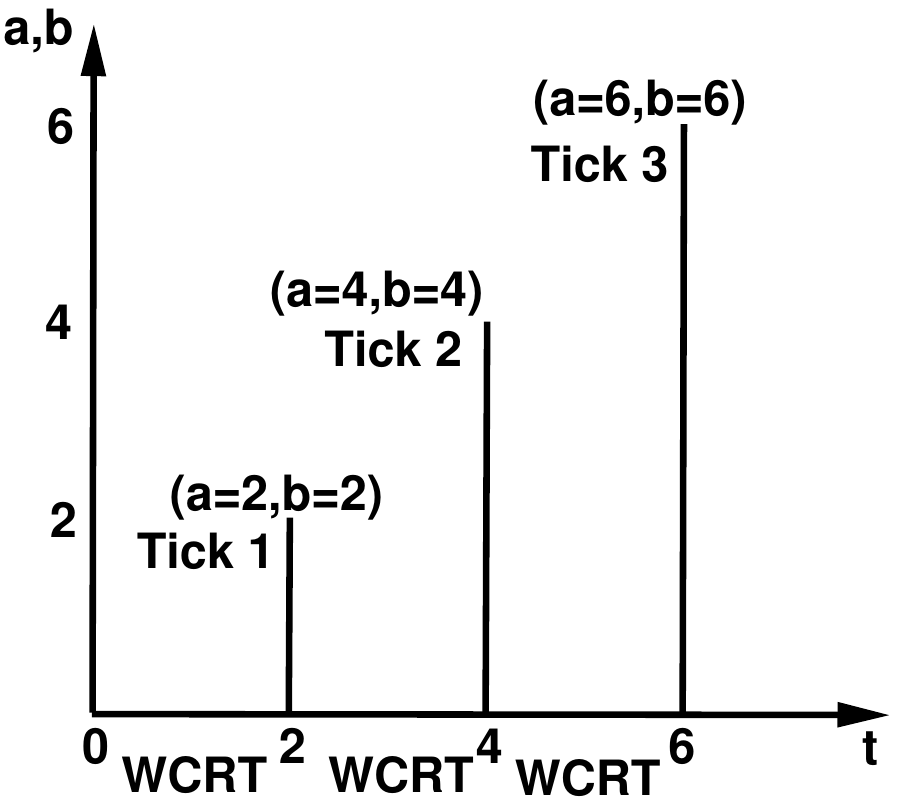}
    \label{fig:14f}}

  \subfloat[Rewrite for flow action in Figure~\ref{fig:12d}]{
    \usebox\rfourth
    \label{fig:14g}}
  \hspace{10pt}
  \subfloat[The timing diagram for Figure~\ref{fig:14g}]{
    \includegraphics[scale=0.25]{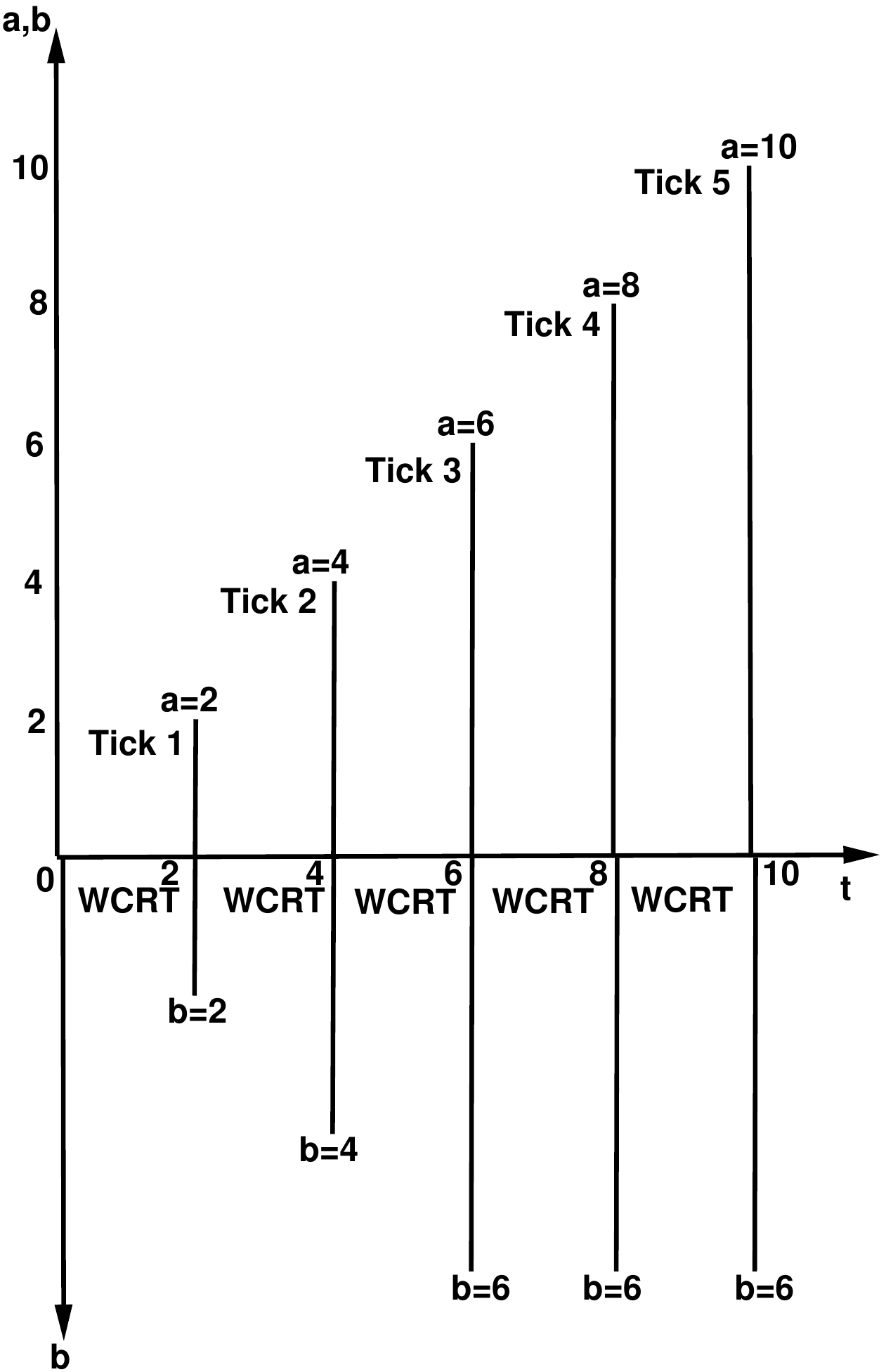}
    \label{fig:14h}}

  \caption{Rewrites for the flow actions in Figure~\ref{fig:12}
    with WCRT = 2}
  \label{fig:14}
\end{figure}

\begin{figure}[t!]
  \ContinuedFloat
  \centering
  \subfloat[Rewrite for flow action in Figure~\ref{fig:12e}]{
    \usebox\rfifth
    \label{fig:14i}}
  \hspace{20pt}
  \subfloat[The timing diagram for Figure~\ref{fig:14i}]{
    \includegraphics[scale=0.4]{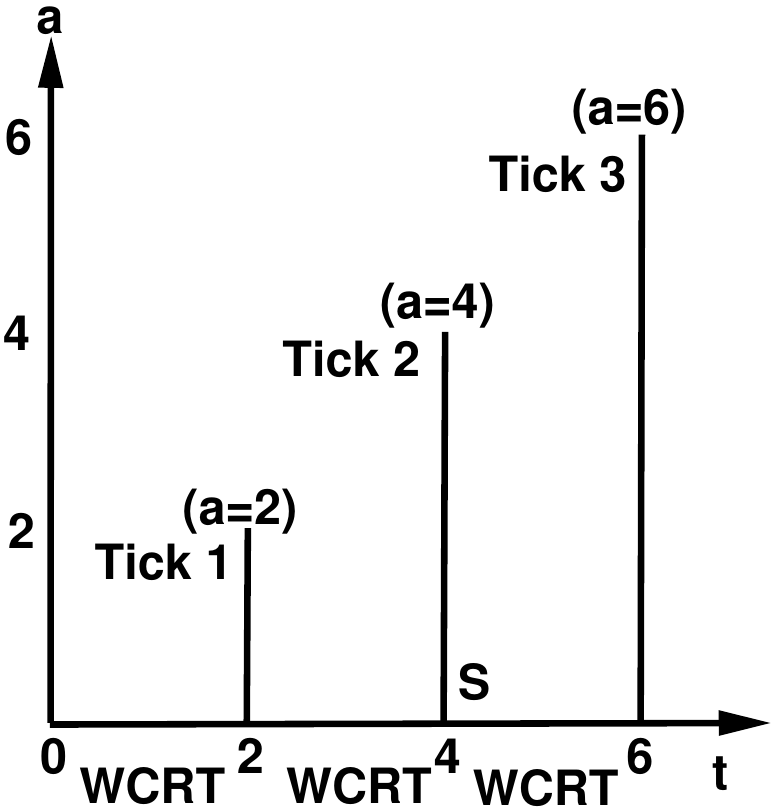}
    \label{fig:14j}}
  \caption{Rewrites for the flow actions in Figure~\ref{fig:12}
    with WCRT = 2}
  \label{fig:14}
\end{figure}

We use the example in Figure~\ref{fig:12a} and its rewrite in
Figure~\ref{fig:14a} to explain how the required \texttt{TTL} algorithm
behavior is achieved. In the rest of the paper we assume for sake of
understanding that the statically computed WCRT value of every HySysJ
program is 2 units. The \texttt{TTL} algorithm (Algorithm~\ref{alg:1})
takes 3 inputs: (1) a list of ODEs ($\Omega$) within one $\mathbf{do}$
block, (2) the $\mathbf{until}$ invariant ($expr$), and (3) the set of
continuous variables evolving in the $\mathbf{do}$ block
($\mathcal{V}$). For our running example these inputs are shown in
Figure~\ref{fig:14a}, line~\ref{lst:8}. Algorithm~\ref{alg:1} computes
for each continuous variable from the set $\mathcal{V}$ its value two
ticks from now (Algorithm~\ref{alg:1}, line~\ref{alg:1:l1}) and places
it into a set $\Delta$. Finally, \texttt{TTL} checks if the values in
set $\Delta$ satisfy the invariant conditions. In the running example,
\texttt{TTL}, when called on the program transition $[0,1)$ obtains the
$[\![a]\!]=0$, the current value of \texttt{a} as 0 (since \texttt{a} is
initialized to 0). The \texttt{filter} function (Algorithm~\ref{alg:1},
line~\ref{alg:1:l1}) first gets the ODE corresponding to variable
\texttt{a}, in this case $a' = 1$. Next, the \texttt{get\_rho} function
gets the rate of change from the ODE, which is simply 1 in this
case. Thus, the computed $\tau$ value is $\tau = 2 * 1 * 2 = 4$
(assuming WCRT = 2). Thus, $a[2]=4$ does not satisfy the invariant
$a\leq 2$ and hence, signal \texttt{R} is emitted in the transition
$[0,1)$ itself. The resultant timing behavior of the rewrite in
Figure~\ref{fig:14a} is shown in Figure~\ref{fig:14b}.

\paragraph{\textit{Semantics of complex flow actions}}
\label{sec:semant-compl-flow}

Multiple continuous variables evolving together can also be handled by
the rewrites. The general rewrite for flow actions evolving multiple
variables in shown in Figure~\ref{fig:7b}. The basic idea of bounded
reduction remains the same. The only difference is that each evolving
variable is reduced sequentially one after the other (line-\ref{lst:13},
Figure~\ref{fig:7b}).

Take for example, the flow action depicted in Figure~\ref{fig:12b}. This
flow action is read as follows: continuous variables \texttt{a} and
\texttt{b} should evolve \textit{simultaneously} (hence, the $||$
composition inside the $\mathbf{do}$ block) until the invariant
condition holds true.  The rewrite and the timing behavior for this
HySysJ program is shown in Figure~\ref{fig:14c} and
Figure~\ref{fig:14d}, respectively. Variables \texttt{a} and \texttt{b},
both evolve at twice the speed of the clock. From the $\mathbf{until}$
expression it is clear that \texttt{a} reaches the value of 16 when
$t = 8$, but \texttt{b} reaches the value of 10 at $t=5$. Furthermore,
given that WCRT = 2, the set
$\Delta = \{a \rightarrow 8, b \rightarrow 8\}$ during the program
transition $[0,2)$, but $\{a \rightarrow 12, b \rightarrow 12\}$ in the
program transition $[2,4)$, which does not satisfy the flow action
invariant, in turn emitting signal \texttt{R}, and hence, the program
terminates at tick 2.

Next, we contrast the flow actions in Figures~\ref{fig:12c}
and~\ref{fig:12d} to show the difference between the synchronous
composition within the $\mathbf{do}$ block and the synchronous
composition of two $\mathbf{do}$ $\mathbf{until}$ blocks. Both these
code snippets have the same ODE expressions. In both these cases,
variables \texttt{a} and \texttt{b} evolve linearly and
\textit{simultaneously}. However, a single invariant condition
constraints the evolution of variables \texttt{a} and \texttt{b} in
Figure~\ref{fig:12c}, while different invariant conditions constraint
the evolution of \texttt{a} and \texttt{b} in Figure~\ref{fig:12d}. The
loop in Figure~\ref{fig:14e} (the rewrite for Figure~\ref{fig:12c}) gets
preempted, in turn terminating the program, at the third logical
tick. In case of Figure~\ref{fig:14g}, only the second synchronous
parallel reaction gets terminated at the third logical tick, but the
whole program cannot terminate due to lockstep semantics of the $||$
operator. Hence, the program terminates only when the loop in the first
synchronous parallel reaction terminates: at tick 5. Variable \texttt{b}
stops evolving at the end of the third tick, whereas \texttt{a} evolves
until the end of the program and takes the value 10.

Until now we have only looked at examples where the evolution of
continuous variables is constrained by invariant conditions akin to the
hybrid automaton. Now, we look at an example where evolution of a
continuous variable is interrupted by a preemption construct. Consider
the code snippet in Figure~\ref{fig:12e}; the flow action states that
\texttt{a} should evolve linearly with the driving clock
\textit{forever}. This continuous action is encapsulated inside an
$\mathbf{abort}$ construct~\footnote{We allow continuous actions to be
  encapsulated in any base language construct.} that preempts the
evolution of \texttt{a} when signal \texttt{S} is present. Signal
\texttt{S} is emitted from a synchronous parallel reaction in the second
logical tick. The rewrite for this code snippet is shown in
Figure~\ref{fig:14i} along with its timing behavior in
Figure~\ref{fig:14j}. The $\mathbf{until(true)}$ invariant condition is
converted into a simple \mbox{$\mathbf{loop}$ $\textbf{pause}$} blocking
condition. The rest of the program remains the same. Signal \texttt{S}
is emitted in the second tick, and responded to by the $\mathbf{abort}$
construct in the third tick, due to delayed signal semantics. The final
observable value of \texttt{a} is 6 when the program terminates. This
preemption based termination of continuous actions will be an important
component of modeling time-delayed mode-switches.

\subsubsection{Schizophrenia}
\label{sec:schizophrenia}

Encapsulating flow actions within preemption statements instantaneously
brings forth the question of schizophrenia -- the possibility that a
single continuous variable can take different values in the same logical
tick.

\newbox{\schtwo}
\begin{lrbox}{\schtwo}
  \begin{lstlisting}[mathescape,style=sysj,morekeywords={until,wait_inbetween,
      cont,signal,loop,abort,await,emit,present,trap,pause,exit,delay,suspend}]
    cont a op+ = 1;
    input signal FAULT;
    loop {
      abort (FAULT) {
        int signal S op+ = 0;
        signal R;
        abort (R) 
        loop {
          ?S = ?S + 1;
          a = a + WCRT;
          // value of d = 2
          if (?S == (2-1)) emit R;
          pause;
        }
      }
      a = 1; //resetting a to 1
    }
  \end{lstlisting}
\end{lrbox}

\newbox{\schone}
\begin{lrbox}{\schone}
  \begin{lstlisting}[mathescape,style=sysj,morekeywords={until,wait_inbetween,
      cont,input,signal,loop,abort,await,emit,present,trap,pause,exit,delay,suspend}]
    cont a op+ = 1;
    input signal FAULT;
    loop {
      abort (FAULT) {
        do {a$'$ = 1} until (a <= 5)
      }
      a = 1; //resetting a to 1
    }
  \end{lstlisting}
\end{lrbox}

\begin{figure}[t!]
  \centering
  \subfloat[Example of schizophrenic code snippet]{
    \usebox\schone
    \label{fig:17a}}
  \hspace{20pt}
  \subfloat[The timing diagram for Figure~\ref{fig:17a}. \texttt{FAULT}
  occurs in the first logical tick.]{
    \includegraphics[scale=0.35]{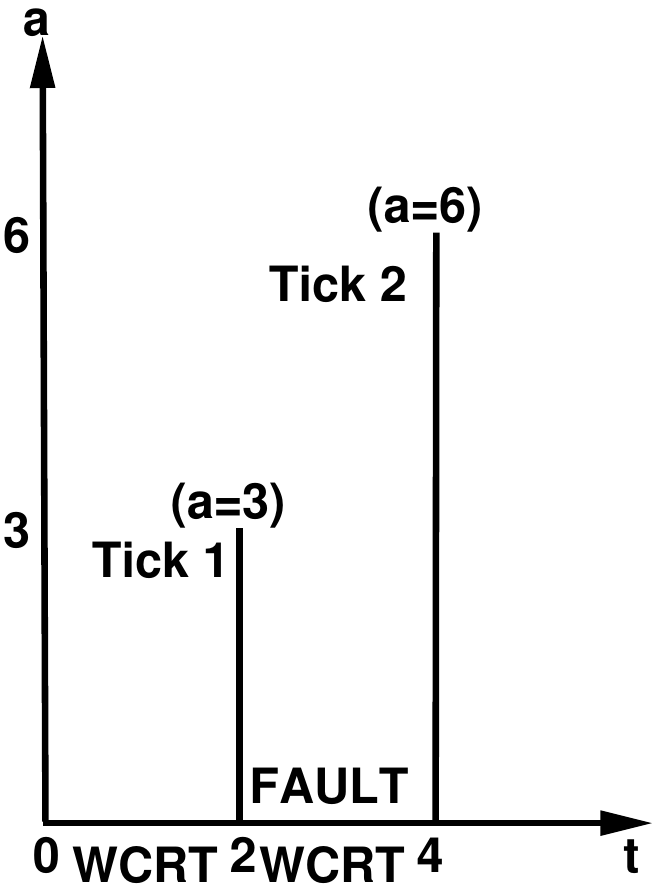}
    \label{fig:17c}}
  
  \subfloat[The timing diagram for Figure~\ref{fig:17a}.
  \texttt{FAULT} signal never occurs]{
    \includegraphics[scale=0.35]{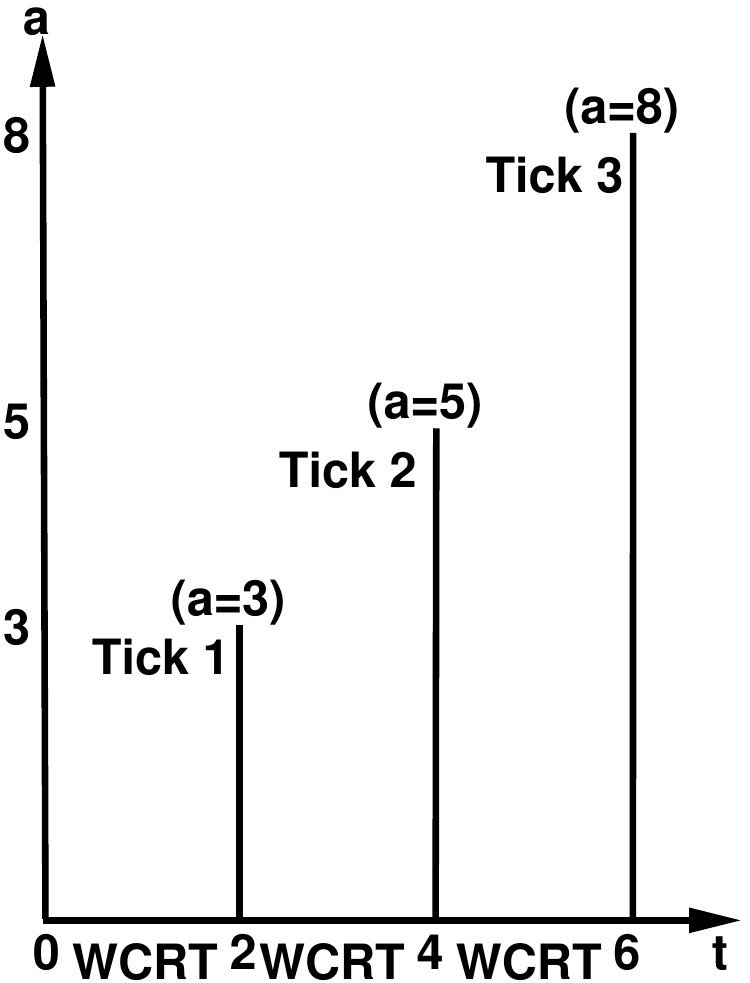}
    \label{fig:17d}}
  \hspace{20pt}
  \subfloat[Correct behavior after insertion of $\mathbf{pause}$ after
  the assignment statement \texttt{a = 1}]{
    \includegraphics[scale=0.35]{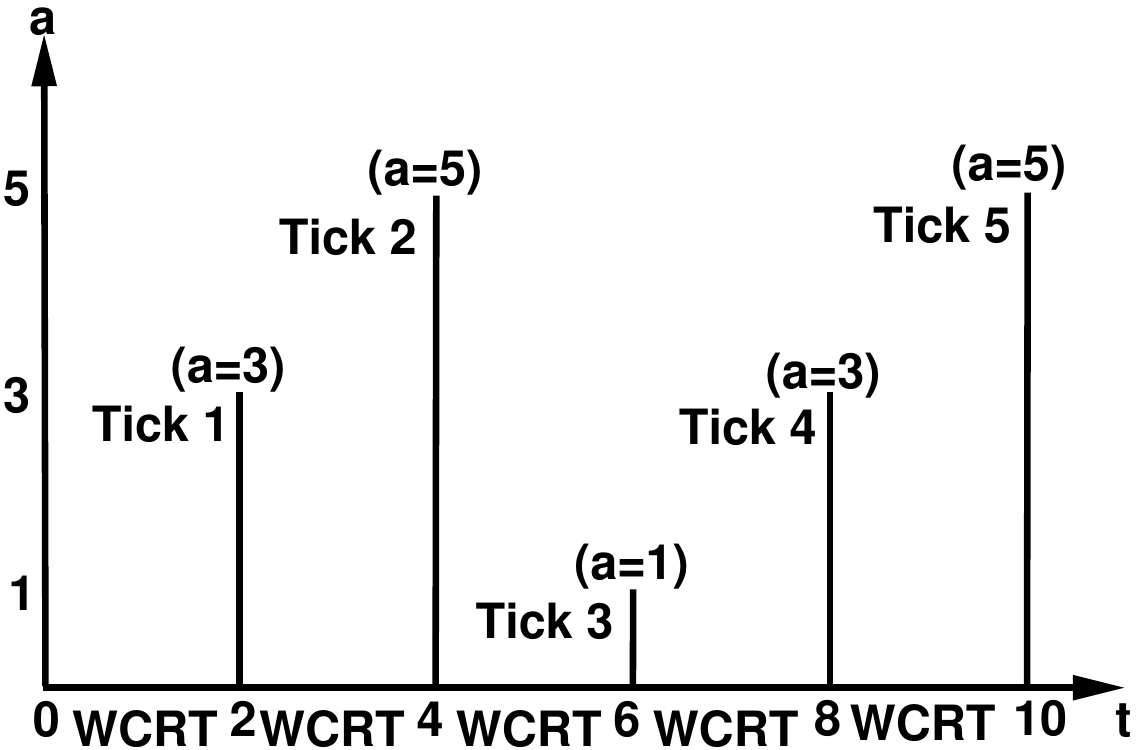}
    \label{fig:17e}}
  \caption{Schizophrenic flow actions in HySysJ}
  \label{fig:17}
\end{figure}

Consider the example code snippet in Figure~\ref{fig:17a}. The
continuous variable \texttt{a} evolves until it reaches 5. This
evolution might be preempted if a \texttt{FAULT} signal is present from
the environment. Once the evolution of the variable is completed or
preempted, \texttt{a} is reset and then evolution begins again, at least
that is the expectation. Given that the initial value of \texttt{a} is 1
and it needs to evolve until it reaches the value of 5 (assuming WCRT =
2) the loop is bounded by two ticks.

The timing behavior of Figure~\ref{fig:17a} is shown in
Figure~\ref{fig:17c}. Let us assume that the \texttt{FAULT} signal does
occur in the first logical tick. Due to delayed signal semantics, the
$\mathbf{abort}$ statement responds to the \texttt{FAULT} signal only in
the second tick. Thus, at the end of the first tick, \texttt{a} takes
the value 3. In the second program transition, $[2,4)$, the evolution of
the variable \texttt{a} stops due to preemption and \texttt{a} is
assigned the value 1. But, due to the $\mathbf{loop}$ statement (line~3)
the program control flow reenters the $\mathbf{do}$ $\mathbf{until}$
block, thereby evolving \texttt{a} again in the same program
transition. Thus, in the second program transition \texttt{a} has two
values: 1 due to the instantaneous assignment statement and \texttt{a =
  a + WCRT} do to reenterance into the flow action. Effectively,
\texttt{a} takes two different values \textit{simultaneously}, this is
termed schizophrenia.

The associative and commutative combination operators defined during
continuous variable declaration are used to resolve such schizophrenic
behavior. In Figure~\ref{fig:17a}, \texttt{a} is declared with the
combination operator $\mathbf{op+}$. This means, if \texttt{a} takes two
different values in the same program transition, then the two values
need to be added together and the result is the final value of
\texttt{a} at the end of the tick. For the program in
Figure~\ref{fig:17a}, during the second program transition, \texttt{a}
takes on two different values: 1 and \texttt{a = 3 + 2 = 5}. Recall that
\mbox{WCRT = 2} and \texttt{a} has the value 3 from the previous
tick. These two values are combined, via addition, together to give the
final result of 6 at the end of the second logical tick as shown in
Figure~\ref{fig:17c}.

The timing behavior of the program in Figure~\ref{fig:17a} without fault
is shown in Figure~\ref{fig:17d}. This again is not the expected
behavior, since the designer expects to reset \texttt{a} once it reaches
the value 5. Thus, the expected value of \texttt{a} is 3 at the end of
the third tick, but the actual value is 8. This unexpected behavior
stems from the fact, that even without faults, the program reenters the
$\mathbf{do}$ $\mathbf{until}$ block due to the $\mathbf{loop}$
statement in the third program transition -- $[4,6)$. Hence, instead of
resetting the value of \texttt{a} to 1, \texttt{a} takes two different
values: 1 and 7 simultaneously, which get combined via the addition
operator to get the final result of 8. Thus, unlike in hybrid automaton,
where reset actions instantaneously reset the value of continuous
variables, resetting the continuous variable requires insertion of the
$\mathbf{pause}$ construct after the assignment statement in HySysJ. The
behavior with insertion of a $\mathbf{pause}$ statement after assignment
statement \texttt{a = 1} is shown in Figure~\ref{fig:17e}.

\subsubsection{Write-write semantics}
\label{sec:write-write-semant-2}

Writing simultaneously to the same continuous variable in the same
program transition is allowed in HySysJ. Writing simultaneously can be
achieved via synchronous parallel composition and these simultaneous
writes to the same continuous variable are resolved using the same
technique (combination operators) as described in
Section~\ref{sec:schizophrenia}. Two examples of simultaneous writes are
shown in Figures~\ref{fig:18a} and~\ref{fig:18c}. The timing behavior
(see Figure~\ref{fig:18b}) is as expected in case of
Figure~\ref{fig:18a}. In the first tick, \texttt{a} takes the value 3,
even though the initial condition (and subsequent reset value) is 0, due
to the combination operator $\mathbf{op+}$. 

\newbox{\swtwo}
\begin{lrbox}{\swtwo}
  \begin{lstlisting}[mathescape,style=sysj,morekeywords={until,wait_inbetween,
      cont,input,signal,loop,abort,await,emit,present,trap,pause,exit,delay,suspend}]
    cont a op+ = 0;
    do {a$'$ = 1 || a$'$ = 1} until (a <= 4)
  \end{lstlisting}
\end{lrbox}

\newbox{\swone}
\begin{lrbox}{\swone}
  \begin{lstlisting}[mathescape,style=sysj,morekeywords={until,wait_inbetween,
      cont,input,signal,loop,abort,await,emit,present,trap,pause,exit,delay,suspend}]
    cont a op+ = 0;
    loop {
      do {a$'$ = 1} until (a <= 4)
      ||
      {a = 1; pause};
      a = 0; //resetting to 0.
      pause
    }
  \end{lstlisting}
\end{lrbox}

\begin{figure}[t!]
  \centering
  \subfloat[Example of simultaneous writes to the same continuous
  variable in HySysJ. WCRT= 2]{
    \usebox\swone
    \label{fig:18a}}
  \subfloat[The timing diagram for Figure~\ref{fig:18a}]{
    \includegraphics[scale=0.35]{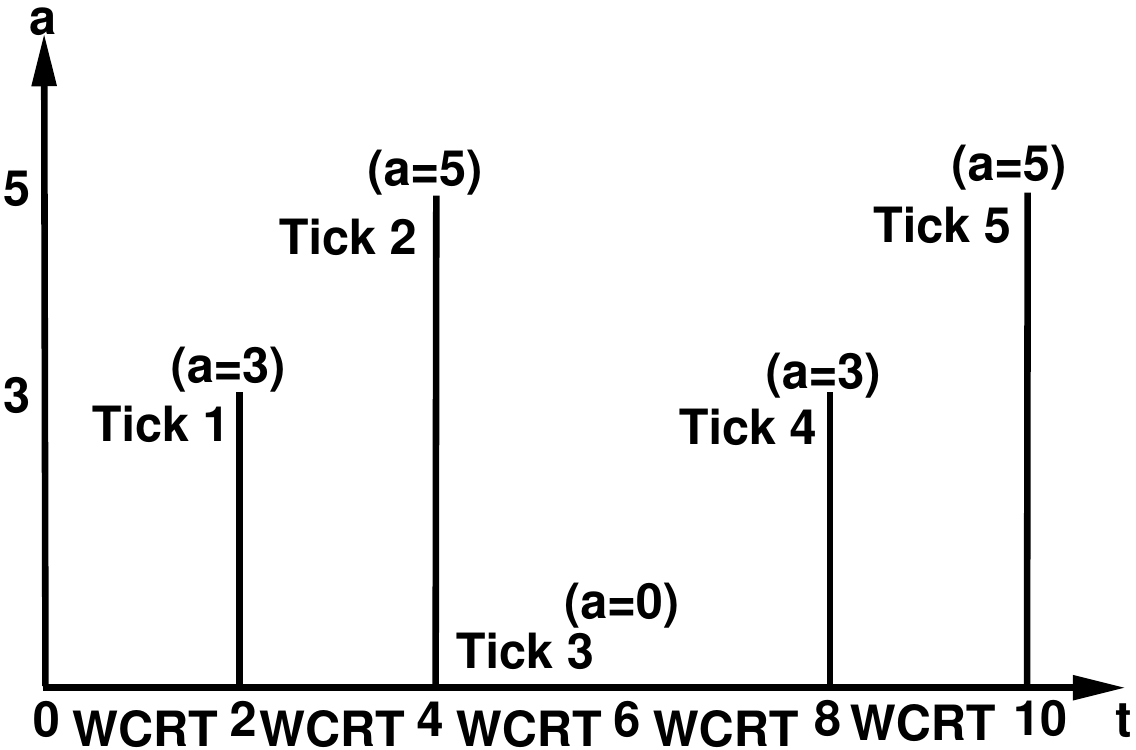}
    \label{fig:18b}}
  
  \subfloat[Example of simultaneous writes to the same continuous
  variable in HySysJ. WCRT=2]{ \usebox\swtwo
    \label{fig:18c}}
  \hspace{20pt}
  \subfloat[The timing diagram for Figure~\ref{fig:18c}]{
    \includegraphics[scale=0.4]{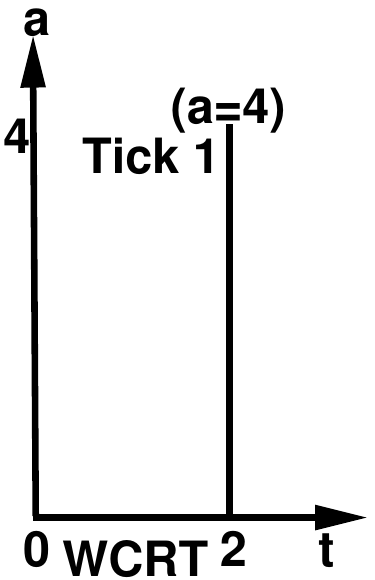}
    \label{fig:18d}}
  \caption{Write-write semantics in HySysJ}
  \label{fig:18}
\end{figure}

\begin{algorithm}[t!]
  \begin{minipage}{1.0\linewidth}
    \SetAlgoLined
    \SetKwInOut{Input}{Input}
    \Input{$\Omega$: a list of ODEs from one $\mathbf{do}$ block}
    \Input{$expr$: the $\mathbf{until}$ expression}
    \Input{$\mathcal{V}$: the set of continuous variables in
      $\Omega$}
    \Input{$\mathcal{M}$: the map from continuous variable to combine
      operator~\footnote{One can statically check at compile time that
        the combine operator is linear}}
    \KwResult{a Boolean value}
    let $\Delta$ $\leftarrow$ $\emptyset$\;
    \For {each $v$ in $\mathcal{V}$} {
      \tcp{$\mathcal{R}$ contains the rate of change}
      $\mathcal{R}$ $\leftarrow$ get\_rhos (filter ($\Omega$, $v$))\;
      \uIf {$|\mathcal{R}| > 1$} { \label{alg:2:l1}
        $\Gamma$ $\leftarrow$ (map ($\lambda\_ \rightarrow [\![v]\!]$)
        $\mathcal{R}$) \label{alg:2:l2}\;
        \tcp{Compute the value of $v$ two ticks from now}
        \For {i in 0..1} {
          \tcp{$\gamma$ is the current value of $v$}
          \tcp{$\rho$ is the rate of change of $v$}
          $\tau$ $\leftarrow$ reduce ($\mathcal{M}.get(v)$,\\
          (map ($\lambda\rho.\lambda\gamma \rightarrow
          \gamma + \rho * WCRT$) $\mathcal{R}$ $\Gamma$)) \;
          $\Gamma$ $\leftarrow$ (map ($\lambda\_ \rightarrow [\![\tau]\!]$) $\Gamma$)\;
          $i$ $\leftarrow$ $i+1$\;
        }
        $\tau$ $\leftarrow$ reduce($\mathcal{M}.get(v)$, $\Gamma$)\;
        $\Delta$ $\leftarrow$ $\Delta \cup \{v \rightarrow
        \tau\}$ \label{alg:2:l3}\;
      }
      \Else {
        $\tau$ $\leftarrow$ $[\![v]\!]$ + $2$ * $\mathcal{R}.get(0)$  * WCRT\;
        $\Delta$ $\leftarrow$ $\Delta \cup \{v \rightarrow \tau\}$\;
      }
    }
    \Return holds\_at\_delta ($expr$, $\Delta$)\;
    \caption{New algorithm to calculate TTL}
    \label{alg:2}
  \end{minipage}
\end{algorithm}

HySysJ also allows simultaneous writes within the same $\mathbf{do}$
blocks as shown in Figure~\ref{fig:18c}. The algorithm
(Algorithm~\ref{alg:1}) needs to be modified now that simultaneous
writes to the same variable are allowed within the same $\mathbf{do}$
block. The new \texttt{TTL} procedure is shown in Algorithm~\ref{alg:2}.

In the modified version, a new input argument is required: a map from
the continuous variable to its corresponding combine operator. The
overall result of Algorithm~\ref{alg:2} is the same Boolean value as
Algorithm~\ref{alg:1}, except, that combine operators are now invoked to
calculate the value of all continuous variables, two logical ticks from
the current tick, which evolve simultaneously within the same flow
action.

The new \texttt{TTL} procedure can be described with the example program
code in Figure~\ref{fig:18c}. Continuous variable \texttt{a} is evolving
simultaneously and linearly until it is less than or equal to 4. In
Algorithm~\ref{alg:2}; $\mathcal{V} = \{a\}$, $\Omega = [a'=1, a'=1]$
and $\mathcal{M} = \{a \rightarrow +\}$. First, Algorithm~\ref{alg:2}
checks if variable $a$ is being modified by more than 1 ODE
simultaneously (line~\ref{alg:2:l1}), which is true in this case. The
then branch is taken and value of $\tau$ is calculated as 12
(lines~\ref{alg:2:l2}-\ref{alg:2:l3}), thus, $\Delta = \{a \rightarrow
12\}$, which does not satisfy the invariant. Hence, the program
terminates at end of the very first program transition.


The resultant timing behavior is shown in Figure~\ref{fig:18d}. In
Algorithm~\ref{alg:2}, we especially need to check that combine operator
is linear, in order to enforce compatibility with linear hybrid
automaton. A programmer might use the $\mathbf{op*}$ (multiplication)
operator to combine simultaneously evolving continuous variables, which
models a higher order ODE.\footnote{Higher order ODEs resulting from
  multiplication operator can be accommodated into HySysJ, but we leave
  this as future work.} Programs combining continuous variables with
non-linear combine operator are rejected at compile time.

\subsubsection{Read-write semantics}
\label{sec:read-write-semantics}

HySysJ allows simultaneous, using synchronous parallel operator $||$,
reading and writing to a single continuous variable. Consider the
example in Figure~\ref{fig:13a}. Continuous variable \texttt{a} evolves
linearly with the driving clock. Simultaneously the value of \texttt{a}
is checked in the $\mathbf{if}$ block. If value of \texttt{a} is between
0 and 2 then signal \texttt{S1} is emitted, else signal \texttt{S2} is
emitted. Furthermore, this branching condition is encapsulated in a
$\mathbf{loop}$. This $\mathbf{loop}$ is preempted once \texttt{a} takes
the value 5. In this case the (assuming again that WCRT = 2) flow action
terminates after 2 ticks.

The timing behavior is shown in Figure~\ref{fig:13b}. In the very first
program transition ($[0,2)$) the value of \texttt{a} is 1. Since the
continuous variables are only updated at the end of the current tick,
the branching condition is satisfied and signal \texttt{S1} is emitted
at the end of the first tick and \texttt{a} takes the value 3. In the
second program transition ($[2,4)$) the branching condition is again
satisfied, again signal \texttt{S1} is emitted and at the end of this
tick \texttt{a} takes the value 5. At the start of the next transition,
the flow invariant does not hold, and hence, in the third transition,
signals \texttt{R} and \texttt{S2} are emitted. Variable \texttt{a} also
stops evolving further. Finally, the program terminates after tick 4,
due to delayed signal semantics.

Simultaneous reading and writing of continuous variables works in HySysJ
due to the delayed semantics. Simultaneous read-write on continuous
variables would need to be rejected if reading the value of a continuous
variable would read the currently evolving value, which is undefined
(and unstable) during the program transition. A continuous variable only
takes a defined (and stable) value at the end of ticks.

\newbox{\iecf}
\begin{lrbox}{\iecf}
  \begin{lstlisting}[mathescape,style=sysj,morekeywords={until,cont,signal,loop,abort,await,emit,present,trap,pause,exit,delay,suspend}]
    signal S1, S2, R;
    cont a = 1;
    abort (R) {
      {do {a$'$ = 1} until (a <= 5); emit R} 
      || 
      {loop {if (a>=0 && a <= 2) emit S1 else emit S2; pause}}
    }
    pause
  \end{lstlisting}
\end{lrbox}

\newbox{\iecs}
\begin{lrbox}{\iecs}
  \begin{lstlisting}[mathescape,style=sysj,morekeywords={until,cont,signal,loop,abort,await,emit,present,trap,pause,exit,delay,suspend}]
    cont a = 0;  a = a + 1;
  \end{lstlisting}
\end{lrbox}

\begin{figure}[t!]
  \centering \subfloat[Example of simultaneous read-write in HySysJ]{\usebox\iecf \label{fig:13a}}

  \subfloat[Timing behavior for Figure~\ref{fig:13a}] {
    \includegraphics[scale=0.35]{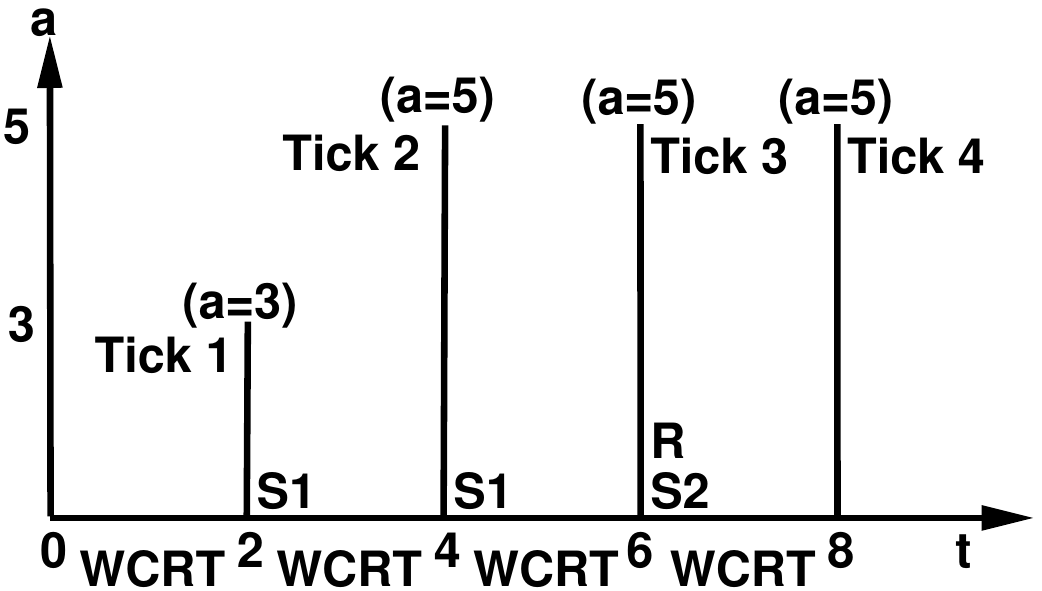}
    \label{fig:13b}
  }
  \caption{Example of read-write semantics on continuous variables in
    HySysJ. Assume that WCRT=2}
  \label{fig:13}
\end{figure}



\section{Determining the value of WCRT for the discrete plant model}
\label{sec:determ-value-wcrt}

The rewrite semantics approximate the plant model in the discrete-time
domain. This raises the question -- \textit{what should be the value of
  WCRT?}

The best approximation would be to allow WCRT to \textit{approach zero},
which is equivalent to performing a definite integration on a continuous
function, representing the plant, as is done in the hybrid
automaton. Another approach is to use zero crossings\footnote{Roughly
  speaking, zero crossing is an event occuring during the integration of
  an ODE, when some expression changes sign from negative to positive.}
and using non-standard analysis as is done in hybrid data-flow
languages~\cite{DBLP:conf/hybrid/BourkeP13}, \cite{simulink}. But, both
these approaches do not (and cannot) consider the time taken for
discrete control transitions. We take a different approach to
determining the value of WCRT, which is tightly related to the
definition of \textit{observability} in classical supervisory control
theory.

Consider a linear, time invariant (LTI), discrete-time
plant\footnote{Every hybrid automaton models a liner time-invariant
  plant in the continuous-time domain. But, now that we have discretized
  the plant, we can use a linear discrete-time invariant system.} in the
state space form as shown in Equation~(\ref{eq:2}). The status of the
plant as observed by the discrete controller is shown in
Equation~(\ref{eq:3})~\footnote{In our case, due to delayed semantics,
  Equation~\ref{eq:3} is actually: \mbox{$\mathbf{y}(k+1) = \mathbf{C}_d
    \mathbf{x}(k)$}}, where $\mathbf{x}(k) \in \mathbb{R}^n$,
$\mathbf{y}(k) \in \mathbb{R}^p$, $n$ and $p$ are the length of the
$\mathbf{x}$ and $\mathbf{y}$ vectors, respectively. $\mathbf{A}_d$ and
$\mathbf{C}_d$ are constant matrices of appropriate dimensions. Then the
observability matrix $\mathcal{O}(\mathbf{A}_d, \mathbf{C}_d)$ is
defined in Equation~\ref{eq:5}. Classical control theory states that one
can learn \textit{everything} about the dynamical behavior of the plant
by using only the observability matrix with the condition that the rank
of $\mathcal{O}$ is $n$.

\begin{scriptsize}
\begin{align}
   \mathbf{x}(k + 1) & = \mathbf{A}_d\mathbf{x}(k), \hspace{10pt} \mathbf{x}(0) =
   \mathbf{x}_0 
   \label{eq:2} \\
   \mathbf{y}(k) & = \mathbf{C}_d\mathbf{x}(k)
   \label{eq:3} \\
   \mathcal{O}(\mathbf{A}_d, \mathbf{C}_d) & =
   \left[
   \begin{array}{ccccc}
     \mathbf{C}_d &
     \mathbf{C}_d\mathbf{A}_d& 
     \mathbf{C}_d\mathbf{A}_d^2 &
     \ldots&
     \mathbf{C}_d\mathbf{A}_d^{n-1}
   \end{array}
   \right]^T
   \label{eq:5}
\end{align}
\end{scriptsize}

The definition of observability matrix is obtained via equating the
inductive definition of the plant model in Equation~(\ref{eq:2}) to the
observed outputs in Equation~(\ref{eq:3}).  Thus, this derivation of the
so called observability matrix holds \textrm{iff} the time taken by the
discrete control transition is equal to the resolution of the discrete
plant model. Hence, \textit{WCRT of plant is equal to the WCRT of the
  controller.}

Intuitively, plant does not really have a WCRT, since it is a continuous
function. We discretize the plant with the WCRT value equal to the one
calculated for the controller, independent of the plant model, in order
to adhere to the classical LTI discrete-time supervisory control theory.

\section{The manufacturing system revisited}
\label{sec:manuf-syst-revis}

\newbox{\fman}
\begin{lrbox}{\fman}
  \lstset{escapeinside={(*@}{@*)}} 
  \begin{lstlisting}[mathescape,style=sysj,morekeywords={until,cont,signal,loop,abort,await,emit,present,trap,pause,exit,delay,suspend}]
    {
      // The controller
      int signal S1, S2, S3 op+ = 0;
      signal DONE;
      loop {
        abort (S1) loop A: pause; (*@\label{lA}@*)
        if (?S1 == 1) { 
          ?S2 = 1; emit S2; 
          abort(DONE) loop B: pause (*@\label{lB}@*)
        };
        else {
          ?S3 = 1; emit S3; 
          abort(DONE) loop C: pause (*@\label{lC}@*)
        }
        ?S2 = 0; ?S3 = 0;
      }
    } || {
      // The plant
      cont x ,y; signal ERROR;
      loop {
        do {x$'$ = 1} until (x <= $\alpha$); (*@\label{d1}@*)
        if (x == $\alpha$) { (*@\label{l1}@*) 
          ?S1 = 1; emit S1;
          abort (S2 || S3) 
           do {x$'$ = 1} until (true);(*@\label{d2}@*)
          if(S2) 
           do {x$'$ = 1 || y$'$=1} until (y <= $\theta$)
          else 
           do {x$'$ = 1 || y$'$=-1} until (y >= 0);
          if (x < $\beta$) {
           do {x$'$ = 1} until (x <= $\beta$);
           x = 0; emit DONE;
          } else emit ERROR;
        } else emit ERROR; pause
      }
    }
  \end{lstlisting}

\end{lrbox}
\begin{figure}[t!]
  \subfloat[The manufacturing system in HySysJ] {
    \usebox\fman \label{fig:19a}}
  \subfloat[Timing behavior with $\alpha=3$, WCRT=2]{
    \includegraphics[scale=0.3]{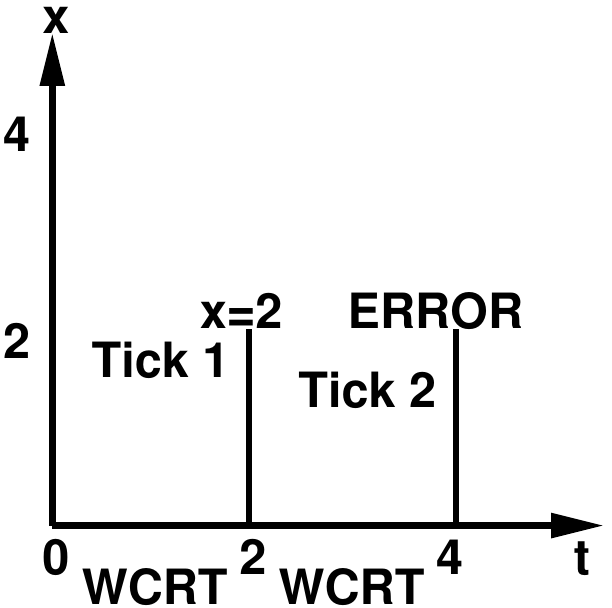}
    \label{fig:19b}}
  
  \subfloat[Timing behavior with $\alpha=2$, WCRT=2]{
    \includegraphics[scale=0.26]{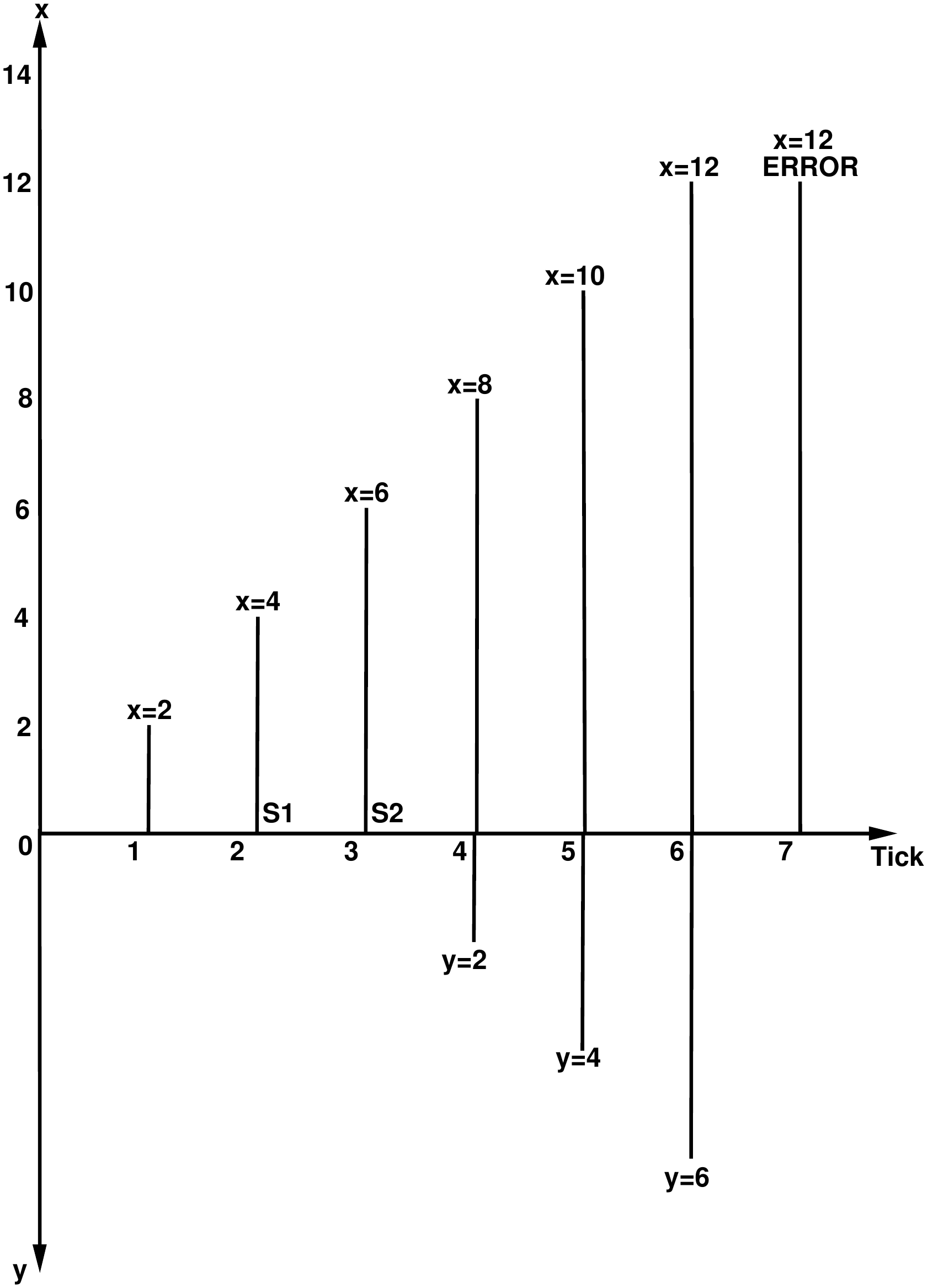}
    \label{fig:19c}}
  \subfloat[Timing behavior with $\alpha=1$, WCRT=1]{
    \includegraphics[scale=0.3]{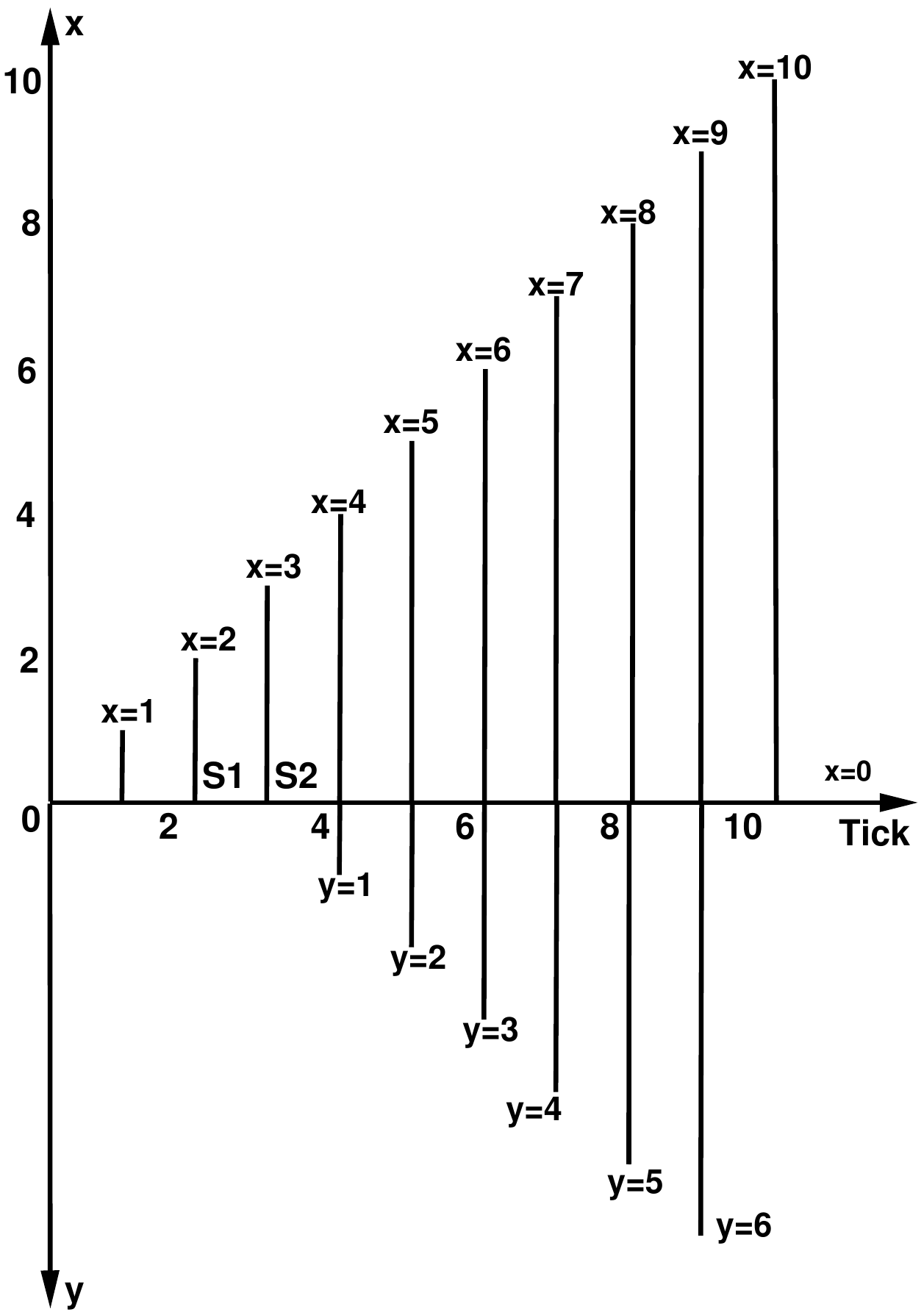}
    \label{fig:19d}}
  \caption{The manufacturing system implemented in HySysJ and its timing
    behavior}
  \label{fig:19}
\end{figure}

We can now revisit the manufacturing control system described in
Section~\ref{sec:motivating-example} and design it in
HySysJ. Furthermore, we verify two properties that are violated in the
hybrid automaton model of this closed loop control system: (1) the
\texttt{TRDC} controller is placed in the correct position so that it
can always observe the passage of the ice-cream on the first carousel
and (2) the non-zero mode-switch time is correctly accounted for in the
control system so that the ice-cream is routed to the correct
storage. The first property is related to the \textit{Observability}
criteria in the classical LTI discrete-time supervisory control theory
and the second is its dual; the \textit{Controllability} criteria. We
will emit an \texttt{ERROR} signal if either of the property is
violated. The verification tool then simply needs to guarantee that
there is no path in the system that reaches the state with emission of
the \texttt{ERROR} signal. This reachability test can be performed on a
symbolic transition system generated from the base SystemJ language, as
all HySysJ statements are rewritten into SystemJ, based on the formal
semantics presented in Appendix~\ref{sec:symb-trans-syst}.

\subsection{Synchronous parallel composition of the plant and the controller}
\label{sec:synchr-parall-comp}

Figure~\ref{fig:19a} shows the HySysJ program implementing the closed
loop control system. There are two synchronous parallel reactions: the
first is the controller and the second is the model of the plant
itself. Before delving into the details, we give an intuitive
justification for a synchronous composition of the plant model and the
controller.

\begin{align}
  \mathbf{x}(k + 1) & = \mathbf{A}_d\mathbf{x}(k) + \mathbf{B}_d
  \mathbf{u}(k), \hspace{30pt} \mathbf{x}(0) = \mathbf{x}_0
   \label{eq:6}
\end{align}

Consider the classical LTI discrete-time control system in
Equation~(\ref{eq:6}) in the state space form. The vector $\mathbf{u}$
takes the control system from some initial state
$\mathbf{x}(0) = \mathbf{x}_0$, $\mathbf{x}(k) \in \mathbb{R}^n$ to some
desired final state $\mathbf{x}(k_1) = \mathbf{x}_f$ in finite number of
time steps $k_1 < \infty$, \textrm{iff} the controllability matrix has
rank $n$. Instead of the controllability criteria, we right now are more
interested in the state transition system as presented in
Equation~(\ref{eq:6}). Observe that the whole control system
(represented by vector $\mathbf{x}$), which includes the controller
state and the plant state always make a transition \textit{together} to
the next state depending upon the current state and the current
input~\footnote{Of course, in our case, the plant responds to previous
  input rather than the current input, i.e., Equation~\ref{eq:6} is
  time-shifted, because of delayed semantics. But, the controllability
  criteria still remains the same.}. This simultaneous transition of the
plant state and the controller state implies a synchronous product (a la
the $||$ composition) of the plant and the controller state transition
systems.

\subsection{Observability}
\label{sec:observability}

We first verify that every ice-cream on the first carousel can be
detected by the \texttt{TRDC} in the manufacturing system from
Section~\ref{sec:motivating-example}. Figure~\ref{fig:19b}, shows the
timing diagram for the control system assuming WCRT=2 and
$\alpha=3$. When the program starts, the controller is in state
\texttt{A} (line~\ref{lA}) waiting for signal \texttt{S1}. The invariant
condition at line~\ref{d1} does not hold after the first tick, and
hence, after \texttt{x} takes the value 2, the $\mathbf{if}$ statement
is checked in the program transition: $[2,4)$. Of course, the
$\mathbf{if}$ condition does not hold (recall that $\alpha = 3$), and
the \texttt{ERROR} signal is generated.

Thus, placing \texttt{TRDC} at 3 units from the start of the first
carousel leads to violation of the observability criteria, a result that
was not detected in the hybrid automaton model. Next, we verify
controllability criteria of our manufacturing system.

\subsection{Controllability}
\label{sec:controllability}

Figures~\ref{fig:19c} and~\ref{fig:19d} show the timing behavior with
$\alpha = 2$/WCRT=2, and $\alpha=1$/WCRT=1ms, respectively. The
observability property is not violated in either case, since the
position of \texttt{TRDC} is exactly divisible the WCRT in both
cases. But, the controllability criteria is violated in
Figure~\ref{fig:19c}.

Upon observing the ice-cream, signal \texttt{S1} is emitted with the
correct TAG value: 1 in Figure~\ref{fig:19a}. The controller responds to
this emission in next tick by emitting signal \texttt{S2}. But, unlike
the hybrid automaton, the ice-cream on the first carousel keeps on
moving and reaches position 6 (in the third tick). This movement of the
ice-cream due to time-delayed mode-switch is modeled on line~\ref{d2},
which can only be preempted by the $\mathbf{abort}$ construct waiting
for emission of signal \texttt{S2} or \texttt{S3}. The rest of the
program behaves similar to the hybrid automaton in
Figure~\ref{fig:1b}. Once the diverter moves 6 arc-length units (recall
that $\theta = 6$), $x$ is already 12, i.e., the ice-cream is past the
end of the first carousel (recall that $\beta=10$) and diverted to the
incorrect storage station. Thereby, again emitting signal \texttt{ERROR}
in tick 7.

A possible configuration that results in correct control behavior is:
placing \texttt{TRDC} at position 1, and with a WCRT = 1 as shown in
Figure~\ref{fig:19d}.

\section{Conclusions}
\label{sec:conclusions}

In this paper we have presented an extension of the linear hybrid
automaton approach to simulation of hybrid systems by including
\textit{non-instantaneous} discrete control transition. Our solution is
to approximate the hybrid model in the discrete domain and preempting
the evolution of the continuous variables at the well established
discrete points in time. We have proposed new constructs in the
synchronous subset of the SystemJ language to model, simulate, and
verify the hybrid systems with non-instantaneous control
transitions. Furthermore, the sound rewrite semantics described in the
paper can be used to build symbolic transition systems, which can be
verified using classical model-checking tools. As a result of this work,
we were able to identify faults in a real hybrid manufacturing control
system that could not be found using simulation of classical linear
hybrid automaton model.

\bibliographystyle{IEEEtran}
\bibliography{/home/amal029/Dropbox/BIBLIOGRAPHY_DATABASE/main_bib}

\begin{thebibliography}{10}
\providecommand{\url}[1]{#1}
\csname url@samestyle\endcsname
\providecommand{\newblock}{\relax}
\providecommand{\bibinfo}[2]{#2}
\providecommand{\BIBentrySTDinterwordspacing}{\spaceskip=0pt\relax}
\providecommand{\BIBentryALTinterwordstretchfactor}{4}
\providecommand{\BIBentryALTinterwordspacing}{\spaceskip=\fontdimen2\font plus
\BIBentryALTinterwordstretchfactor\fontdimen3\font minus
  \fontdimen4\font\relax}
\providecommand{\BIBforeignlanguage}[2]{{%
\expandafter\ifx\csname l@#1\endcsname\relax
\typeout{** WARNING: IEEEtran.bst: No hyphenation pattern has been}%
\typeout{** loaded for the language `#1'. Using the pattern for}%
\typeout{** the default language instead.}%
\else
\language=\csname l@#1\endcsname
\fi
#2}}
\providecommand{\BIBdecl}{\relax}
\BIBdecl

\bibitem{Henzinger:1996:THA:788018.788803}
\BIBentryALTinterwordspacing
T.~A. Henzinger, ``The theory of hybrid automata,'' in \emph{Proceedings of the
  11th Annual IEEE Symposium on Logic in Computer Science}, ser. LICS
  '96.\hskip 1em plus 0.5em minus 0.4em\relax Washington, DC, USA: IEEE
  Computer Society, 1996, pp. 278--. [Online]. Available:
  \url{http://dl.acm.org/citation.cfm?id=788018.788803}
\BIBentrySTDinterwordspacing

\bibitem{Harel:1987:SVF:34884.34886}
\BIBentryALTinterwordspacing
D.~Harel, ``Statecharts: A visual formalism for complex systems,'' \emph{Sci.
  Comput. Program.}, vol.~8, no.~3, pp. 231--274, Jun. 1987. [Online].
  Available: \url{http://dx.doi.org/10.1016/0167-6423(87)90035-9}
\BIBentrySTDinterwordspacing

\bibitem{berry96}
G.~Berry, ``Constructive semantics of {Esterel}: From theory to practice
  (abstract),'' in \emph{AMAST '96: Proceedings of the 5th International
  Conference on Algebraic Methodology and Software Technology}.\hskip 1em plus
  0.5em minus 0.4em\relax London, UK: Springer-Verlag, 1996, p. 225.

\bibitem{nhal91}
N.~Halbwachs, P.~Caspi, P.~Raymond, and D.~Pilaud, ``The synchronous dataflow
  programming language {Lustre},'' \emph{Proceedings of the IEEE}, vol.~79,
  no.~9, pp. 1305--1320, September 1991.

\bibitem{pgue91}
P.~{Le Guernic}, T.~Gautier, M.~{Le Borgne}, and C.~{Le. Marie}, ``Programming
  real-time applications with {SIGNAL},'' \emph{Proceedings of the IEEE},
  vol.~79, no.~9, pp. 1321--1336, September 1991.

\bibitem{clarke-book00}
E.~M. Clarke, O.~Grumberg, and D.~Peled, \emph{Model Checking}.\hskip 1em plus
  0.5em minus 0.4em\relax MIT Press, 2000.

\bibitem{boldt07}
M.~Boldt, C.~Traulsen, and R.~von Hanxleden, ``Worst case reaction time
  analysis of concurrent reactive programs,'' \emph{Electronic Notes in
  Theoretical Computer Science}, vol. 203, no.~4, pp. 65--79, June 2008.

\bibitem{proop10}
P.~S. Roop, S.~Andalam, R.~von Hanxleden, S.~Yuan, and C.~Traulsen, ``{Tight
  WCRT analysis of synchronous C programs},'' in \emph{CASES}, 2009, pp.
  205--214.

\bibitem{wilhelm08}
R.~Wilhelm, J.~Engblom, A.~Ermedahl, N.~Holsti, S.~Thesing, D.~Whalley,
  G.~Bernat, C.~Ferdinand, R.~Heckmann, T.~Mitra, F.~Mueller, I.~Puaut,
  P.~Puschner, J.~Staschulat, and P.~Stenstr\"{o}m, ``The worst-case
  execution-time problem---overview of methods and survey of tools,''
  \emph{Trans. on Embedded Computing Sys.}, vol.~7, no.~3, pp. 1--53, 2008.

\bibitem{Astrom:2008:FSI:1816978}
K.~J. Astrom and R.~M. Murray, \emph{{Feedback Systems: An Introduction for
  Scientists and Engineers}}.\hskip 1em plus 0.5em minus 0.4em\relax Princeton,
  NJ, USA: Princeton University Press, 2008.

\bibitem{Carloni:2006:LTH:1166403.1166404}
\BIBentryALTinterwordspacing
L.~P. Carloni, R.~Passerone, A.~Pinto, and A.~L. Angiovanni-Vincentelli,
  ``Languages and tools for hybrid systems design,'' \emph{Found. Trends
  Electron. Des. Autom.}, vol.~1, no. 1/2, pp. 1--193, Jan. 2006. [Online].
  Available: \url{http://dx.doi.org/10.1561/1000000001}
\BIBentrySTDinterwordspacing

\bibitem{1205169}
A.~Vachoux, C.~Grimm, and K.~Einwich, ``Analog and mixed signal modelling with
  systemc-ams,'' in \emph{Circuits and Systems, 2003. ISCAS '03. Proceedings of
  the 2003 International Symposium on}, vol.~3, May 2003, pp.
  III--914--III--917 vol.3.

\bibitem{Pecheux:2006:VVA:2298529.2301129}
\BIBentryALTinterwordspacing
F.~Pecheux, C.~Lallement, and A.~Vachoux, ``Vhdl-ams and verilog-ams as
  alternative hardware description languages for efficient modeling of
  multidiscipline systems,'' \emph{Trans. Comp.-Aided Des. Integ. Cir. Sys.},
  vol.~24, no.~2, pp. 204--225, Nov. 2006. [Online]. Available:
  \url{http://dx.doi.org/10.1109/TCAD.2004.841071}
\BIBentrySTDinterwordspacing

\bibitem{DBLP:conf/hybrid/BourkeP13}
T.~Bourke and M.~Pouzet, ``Z{\'e}lus: a synchronous language with odes,'' in
  \emph{HSCC}, 2013, pp. 113--118.

\bibitem{bertin2001taxys}
V.~Bertin, E.~Closse, M.~Poize, J.~Pulou, J.~Sifakis, P.~Venier, D.~Weil, and
  S.~Yovine, ``Taxys= esterel+ kronos: A tool for verifying real-time
  properties of embedded systems,'' in \emph{Proceedings of the IEEE Conference
  on Decision and Control}, vol.~3, no. EPFL-CONF-185013, 2001, pp. 2875--2880.

\bibitem{alur94}
R.~Alur and D.~L. Dill, ``A theory of timed automata,'' \emph{Theoretical
  Computer Science}, vol. 126, pp. 183--235, 1994.

\bibitem{baldamus2002modifying}
M.~Baldamus and T.~Stauner, ``Modifying esterel concepts to model hybrid
  systems,'' \emph{Electronic Notes in Theoretical Computer Science}, vol.~65,
  no.~5, pp. 35--49, 2002.

\bibitem{bauer2010synchronous}
K.~Bauer and K.~Schneider, ``From synchronous programs to symbolic
  representations of hybrid systems,'' in \emph{Proceedings of the 13th ACM
  international conference on Hybrid systems: computation and control}.\hskip
  1em plus 0.5em minus 0.4em\relax ACM, 2010, pp. 41--50.

\bibitem{alur1993hybrid}
R.~Alur, C.~Courcoubetis, T.~A. Henzinger, and P.-H. Ho, \emph{Hybrid automata:
  An algorithmic approach to the specification and verification of hybrid
  systems}.\hskip 1em plus 0.5em minus 0.4em\relax Springer, 1993.

\bibitem{albert2004comparison}
A.~Albert, ``Comparison of event-triggered and time-triggered concepts with
  regard to distributed control systems,'' \emph{Embedded World}, vol. 2004,
  pp. 235--252, 2004.

\bibitem{amal10}
A.~Malik, Z.~Salcic, P.~S. Roop, and A.~Girault, ``{SystemJ: A GALS language
  for system level design},'' \emph{Elsevier Journal of Computer Languages,
  Systems and Structures}, vol.~36, no.~4, pp. 317--344, December 2010.

\bibitem{simulink}
``{Simulink, Simulation and Model based design},''
  \url{http://www.mathworks.com.au/products/simulink/}.

\end{thebibliography}

%

\begin{IEEEbiography}{Avinash Malik}
\end{IEEEbiography}





\newpage
\appendices

\section{Proof for the rewrite semantics}
\label{sec:disc-rewr-semant}

In this section we discuss the correctness criteria for the rewrites
described in Section~\ref{sec:flow-actions}.

\subsection{Correctness of the flow action rewrite -- discretizing the
  plant}
\label{sec:soundn-flow-acti}

We have given a rewrite semantics for the flow actions in
Figure~\ref{fig:7}. Every flow action is converted into a loop with a
reduction function on the continuous variable. We prove the correctness
of this rewrite using discretization of derivatives in
Lemmas~\ref{lem:1} and~\ref{lem:2}.

\begin{lem}
  Given $a' = \rho$, $a[n+1] = a[n] + \rho \times WCRT$, where $a[n]$ is
  the value of $a$ at tick $n$.
 \label{lem:1}
\end{lem}

\begin{proof}
  \begin{align*}
   \frac{da}{dt} \approx \frac{a(t+\Delta) - a(t)}{\Delta}\\
   \therefore \frac{a(t+\Delta) - a(t)}{\Delta} = \rho \\
  \end{align*}
  $let, t = \Delta\times n$ and writing $a[n] = a(\Delta\times n)$, we
  get:
  \begin{align*}
   a(\Delta \times (n+1)) - a(\Delta \times n) = \Delta
   \times \rho\\
   \therefore a[n+1] = a[n] + \rho \times \Delta\\
  \end{align*}
Finally, in our case, $\Delta = WCRT$ hence:
\begin{align*}
   a[n+1] = a[n] + \rho \times WCRT
\end{align*}
\end{proof}

\begin{lem}
 Given $a[0]$, $a[k] = a[0] + \sum_{n=0}^{k-1} \rho \times WCRT$. 
 \label{lem:2}
\end{lem}

\begin{proof}
  Follows from the inductive definition of \mbox{$a[n+1]$} in
  Lemma~\ref{lem:1}.
\end{proof}

Lemma~\ref{lem:1} gives the approximation of a derivative into the
discrete time domain. Every synchronous program is clock-driven by
definition, and hence, from Lemma~\ref{lem:1}, for any tick $n+1$ the
valuation of the continuous variable $a$, i.e., $a[n+1]$ is dependent
upon the current value $a[n]$. Furthermore, given the initial value
$a[0]$, the value of the continuous variable at some tick $k$ is given
by Lemma~\ref{lem:2}, which is a reduction: a bounded summation on $\rho
\times WCRT$ added to the initial value. Every bounded summation is
written as a bounded loop and hence, the rewrite holds.

Next, is the question about finding the bound (or equivalently $k$) in
Lemma~\ref{lem:2}. In the rewrites, this bound is calculated by the
algorithms computing the \texttt{TTL}. The \texttt{TTL} algorithms are
evaluated at program execution time. In general, one cannot determine
the number of loop iterations (equivalently $k$) at compile time,
because the value of $k$ depends upon the valuation of the continuous
variable, since the invariant conditions are specified on the valuation
of the continuous variables rather than being bounded on time as in
definite integrals. The proposed \texttt{TTL} algorithms only ever
look-ahead 2 logical ticks to bound the loop. We need to show that this
2 tick look-ahead is \textit{sufficient} to guarantee that the proposed
rewrites never violate the invariant conditions.

\begin{lem}
  Given invariant condition ($expr$) of the flow construct holds at
  $a[0]$, it is \textbf{necessary} that $a[0] + \rho \times WCRT$
  satisfies the invariant.
  \label{lem:3}
\end{lem}

\begin{proof}
  Follows from the observation that preemption is always delayed by 1
  tick and hence, the flow action will be executed for at least 1 tick.
\end{proof}

\begin{rem}
  The invariant condition should always hold when we first enter the
  rewrite, i.e., $a[0]$ always satisfies the flow invariant by
  definition. Moreover, the flow action, from Lemma~\ref{lem:3}, will
  always be executed at least once. Every continuous variable is updated
  only at the end of the tick, hence, the WCRT value needs to be small
  enough so that at the end of the first tick, $a[0] + \rho \times WCRT$
  does not violate the flow invariant.
\end{rem}

\begin{thm}
  Given invariant condition ($expr$) of the flow construct holds at
  $a[0]$ it is \textbf{sufficient} to show that invariant does not
  hold at $a[2]$ for the rewrite to be correct.
  \label{thm:1}
\end{thm}

\begin{proof}
  Follows from Lemma~\ref{lem:3} and induction on the structure of the
  rewrite in Figure~\ref{fig:7}. Observer that in the very first
  iteration (program transition from tick 0 to tick 1) of the loop,
  $a[0]$ is the programmer specified initial value or the default value
  of continuous variable $a$. The reduction statement computes the value
  $a[1]$ and updates $a$ with this value at the end of the tick. For the
  next iteration, following structural induction, $a[0]$ is now the
  value $a[1]$ computed in the last tick. Thus, for any loop iteration,
  representing the transition from tick $n$ to $n+1$, $a[0]$ holds
  summation of the past $n-1$ tick values, from the sum in
  Lemma~\ref{lem:2}, in $a[0]$. From Lemma~\ref{lem:3}, we know that
  $a[0] + \rho \times WCRT$ should always hold, and hence, it follows
  that $a[n] + \rho \times WCRT$ should also hold, which means we only
  need to look 2 ticks ahead to bound the temporal loop.

\end{proof}

\end{document}